\documentclass[12pt,draftclsnofoot,onecolumn]{IEEEtran}

\usepackage{cite} 
\usepackage{xspace}
\usepackage{filecontents}

\usepackage{float}
\usepackage{pgf, tikz}

\usepackage{graphicx}

\usepackage[cmex10]{amsmath}

\usepackage{bm}

\usepackage{amssymb} 

\usepackage{array}

\usepackage{subfigure}
\usepackage{caption2} 

\usepackage{amsfonts}
\usepackage{amsthm}   

\usepackage{enumerate} 

\usepackage{extarrows}

\newtheorem{definition}{Definition} 
\newtheorem{theorem}{Theorem} 
\newtheorem{lemma}{Lemma}

\newtheorem{corollary}{Corollary} 
\newtheorem{remark}{Remark}
\usetikzlibrary{patterns,snakes}

\usepackage{enumerate}

\usepackage{extarrows}

\newcommand{\Esen}{\mathcal{E}_{\mathrm{SB}}}
\newcommand{\Psen}{\mathcal{P}_{\mathrm{SB}}}
\newcommand{\Etx}{\mathcal{E}_{\mathrm{TB}}}
\newcommand{\Ptx}{\mathcal{P}_{\mathrm{TB}}}

\newcommand{\Psuc}{P_{\mathrm{suc}}}

\newcommand{\Pout}{P_{\mathrm{out}}}
\newcommand{\Poutk}[1]{\left(P_{\mathrm{out}}\right)^{{#1}}}

\newcommand{\Pois}[2]{\mathrm{Pois}\left(#1,#2\right)}

\newcommand{\teventsen}[1]{t_{\mathrm{SB},#1}}

\newcommand{\teventsuc}[1]{t_{\mathrm{STB},#1}}

\newcommand{\myexpect}[1]{\mathbb{E}\left\lbrace #1 \right\rbrace}

\newcommand{\myprobability}[1]{\mathrm{Pr}\left\lbrace #1 \right\rbrace}

\newcommand{\Tstd}{T_{\mathrm{UA}}}
\newcommand{\aveTstd}{\bar{T}_{\mathrm{UA}}}

\newcommand{\Tuc}{T_{\mathrm{UC}}}
\newcommand{\aveTuc}{\bar{T}_{\mathrm{UC}}}

\newcommand{\UC}{\mathrm{UC}}

\newcommand{\eventsuc}{\Lambda_{\mathrm{suc}}}

\newcommand{\eventfail}{\Lambda_{\mathrm{fail}}}

\newcommand{\Prf}{\mathcal{P}_{\mathrm{tx}}}

\newcommand{\AuthorOne}{Wanchun Liu}
\newcommand{\AuthorTwo}{Xiangyun Zhou}
\newcommand{\AuthorThree}{Salman Durrani}
\newcommand{\AuthorFour}{Hani Mehrpouyan}
\newcommand{\AuthorFive}{Steven D. Blostein}
\newcommand{\ThankOne}{Wanchun Liu, Xiangyun Zhou, Salman Durrani are with the Research School of Engineering, the
	Australian National University, Canberra, ACT 2601, Australia
	(emails: \{wanchun.liu, xiangyun.zhou, salman.durrani\}@anu.edu.au).
	Hani Mehrpouyan is with the Department of Electrical and Computer Engineering, Boise State University, Idaho, USA (email: hani.mehr@ieee.org).
	Steven D. Blostein is with the Department of Electronic and Computer Engineering, Queen's University, Canada (email: steven.blostein@queensu.ca).
	Part of the results of this paper was presented at Globecom'15.
}

\begin{document}
%
\title{Energy Harvesting Wireless Sensor Networks: Delay Analysis Considering Energy Costs of Sensing and Transmission}
\author{\authorblockN{\AuthorOne,~\AuthorTwo,~\AuthorThree,~\AuthorFour, \\and~\AuthorFive\thanks{\ThankOne}}}
\maketitle

\maketitle
\vspace{-1cm}
\begin{abstract}
\vspace{-0.2cm}	

Energy harvesting (EH) provides a means of greatly enhancing the lifetime of wireless sensor nodes. However, the randomness inherent in the EH process may cause significant delay for performing sensing operation and transmitting the sensed information to the sink. Unlike most existing studies on the delay performance of EH sensor networks, where only the energy consumption of transmission is considered, we consider the energy costs of both sensing and transmission. Specifically, we consider an EH sensor that monitors some {status} property and adopts a harvest-then-use protocol to perform sensing and transmission. To comprehensively study the delay performance, we consider two {complementary} metrics and analytically derive their statistics: 
{(i) update age - measuring the time taken from when information is obtained by the sensor to when the sensed information is successfully transmitted to the sink, i.e., how timely the updated information at the sink is, and 
(ii) update cycle - measuring the time duration between two consecutive successful transmissions, i.e., how frequently the information at the sink is updated.}
Our results show that the consideration of sensing energy cost leads to an important tradeoff between the two metrics: more frequent updates result in less timely information available at the sink.

\end{abstract}
\vspace{-0.5cm}
\begin{IEEEkeywords}
	\vspace{-0.3cm}
Energy harvesting, wirelessly powered communications, delay analysis, energy costs of sensing and transmission.
\vspace{-0.2cm}
\end{IEEEkeywords}

\newpage
\section{Introduction}
\textbf{Background:} Energy harvesting (EH) from energy sources in the ambient environment is an attractive solution to power wireless sensor networks (WSNs). The feasibility of powering WSNs by EH from solar, wind, vibration and radio-frequency (RF) signals has been demonstrated in the literature~\cite{EH_Survey,EH_survey_singapore,WPT_survey,kaibin_zhou,krikidis_survey}. If an EH source is periodically or continuously available, a sensor node can in theory be powered perpetually. However, the design of EH WSNs raises several interesting and challenging issues.

\textbf{Design Challenges:} An important design consideration for EH WSNs is the modeling of energy costs. There are three main energy costs in wireless sensors~\cite{Industrial}: (i) energy cost of RF transmission and reception, including idle listening, (ii) energy cost of information sensing and processing, and (iii) energy cost of other basic processing while being active. Generally, the energy cost of other basic processing is much smaller compared to the energy cost of transmission~\cite{sensor_app}. Hence, the majority of the current work on EH WSNs has considered only the energy cost of transmission, while ignoring the energy cost of sensing~\cite{Niyato,Jing_lei}. For some sensors, such as high-rate and high-resolution acoustic and seismic sensors, the energy cost of sensing can actually be higher than the energy cost of transmission, e.g., see~\cite{SensingPower} and references there in. Hence, it is important to accurately model the energy cost of sensing in WSNs~\cite{Mao}.

For WSNs powered by EH from the ambient environment, the energy arrival process is inherently time-varying in nature. These fluctuations in the energy arrival process can be slow or fast and are characterised by its coherence time~\cite{Yates}. For instance, for the case of EH from a solar panel on a clear day with abundant sunshine, the coherence time is on the order of minutes or hours. For the case of wireless energy transfer via RF signals, the coherence time can be on the order of milliseconds, which is comparable to the duration of a communication time slot. The energy arrival process in the latter case can be modeled as a random process where the amount of harvested energy in each time slot follows some probability distribution. For example, papers studying EH from RF signals often assume an exponential distribution~\cite{Tianqing,Ali,Ding}. Another example, using the gamma distribution, can be found in~\cite{Luo}. However, many energy arrival processes in practice cannot be accurately modeled by using exponential or gamma distributions. The consideration of a more general probability distribution for modeling the amount of energy arrival is still largely an open problem.

{In many sensor network applications, the delay performance is a key design challenge.
The effects of randomness in both arrivals of the multiple data packets and harvested energy on the overall transmission completion time were considered in~\cite{Yangjing_delay}. 
In~\cite{Tianqing}, a single data packet and randomness in the energy arrival process and wireless channel, were considered in the analysis of transmission delay, i.e., the time duration between the generation of a packet and its successful reception at the sink.
However, both~\cite{Yangjing_delay} and~\cite{Tianqing} only considered the energy cost of transmission.
To the best of our knowledge, a comprehensive analysis of the delay performance of EH WSNs taking into account a realistic model of sensor energy costs, has not been investigated in the literature.}\\

\textbf{Paper Contributions:} 
{We consider a status \emph{monitoring} scenario, e.g., monitoring some property of a target environment,  with one sensor-sink pair.
The sensor is solely powered by EH from an ambient energy source.
The sensor periodically monitors and senses the current environment, i.e., it generates current status information about one or more variables of interest, and then transmits a status-information-containing packet to the sink\footnote{{Due to the fluctuation in the energy arrival process, strictly periodic sensing and transmission is not possible.
	In this paper, `periodic' is used to indicate that the sensor alternates between sensing and transmission(s) in order to keep status updating at the sink.}}.
Once the packets are successfully transmitted to the sink, which may occur after several failed retransmissions due to fading in the transmission channel, the status under monitoring is \emph{updated} at the sink.

We adopt two different metrics to assess the delay performance:
(i) update age\footnote{{The term update age is inspired by~\cite{R1} and indicates the age or timeliness of the transmitted information, since an outdated message may lose its value in a communication system when the receiver has interest in fresh information \cite{Newage}. Note that} this notion of the delay is in fact the same as the transmission delay in \cite{Tianqing}.} which measures the time duration between the time of generation of the current status information at the sink and the time at which it is updated at the sink, and 
(ii) update cycle which measures the time duration between one status update at the sink to the next.
The update age (or freshness) and update cycle (or frequency) are complementary measures.
For instance, a smaller update age means the updated status information at the sink is much more timely,
but does not indicate when the next update status information will be received.
A smaller update cycle means more frequent status updates at the sink,
but does not indicate when the current updated status information was originally generated or how old it is.
Thus, the quality of a status monitoring system, i.e., the status update freshness and frequency, is comprehensively captured by the update age and update cycle, respectively.}

We account for the fact that sensing and transmission operations both consume energy. Inspired from the harvest-then-use and save-then-transmit communication protocols for EH nodes in wireless networks~\cite{Luo,Ali,Tianqing}, which are simple to implement in practice, we consider a harvest-then-use protocol for the EH sensor. In our proposed protocol, the sensor performs sensing and transmission as soon as it has harvested sufficient energy. In order to {limit} the delay due to retransmissions, we impose a time window for retransmissions. The delay performance of the considered harvest-then-use protocol is analyzed. The main contributions of this paper are as follows:
\begin{itemize}
    \item
    We provide a comprehensive study on the delay performance of EH sensor networks. Apart from the commonly considered delay due to the information transmission from the sensor to the sink, defined as the update age, we also characterize the frequency of updating the information held by the sink, defined as the update cycle.

    \item
    Considering a Rayleigh fading wireless channel, we analytically derive the statistics of both the update cycle and the update age. We consider both a deterministic energy arrival model and a random energy arrival model with a general distribution, so that our results can be applied to model a wide range of EH processes.

	\item
    We take the energy costs of both sensing and transmission into account when studying the delay performance. Such a consideration brings up an interesting question of whether to increase or reduce the number of allowed retransmission attempts for each sensed information, because both sensing and transmission consume energy. This in turn results in a tradeoff between the update cycle and the update age.  The tradeoff emphasizes the importance of modeling the energy cost of sensing.
\end{itemize}


\textbf{Notations:} $\myexpect{\cdot}$ and $\myprobability{\cdot}$ are expectation and probability operators, respectively. Convolution operators for continuous and discrete functions are denoted as $\star$ and $\ast$, respectively. $\lceil \cdot \rceil$ and $\lfloor \cdot \rfloor$ are ceiling and floor operators, respectively. $\sum_{i=m}^{n}$ is the summation operator, and if $m>n$, the result is zero. $\Pois{i}{\lambda}$ is the probability mass function (pmf) of a Poisson distribution with parameter $\lambda$. 
%
%
\section{System Model}
We consider the transmission scenario where a sensor periodically transmits its sensed information to a sink, as illustrated in Fig.~\ref{fig:0}. The sensor is an EH node which harvests energy from the ambient environment such as solar, wind, vibration or RF signals. The sensor has two main functions, i.e., sensing and transmission, each having individual energy cost. We assume half-duplex operation, i.e., sensing and transmission cannot occur at the same time. In order to perform either sensing or transmission, the sensor first needs to spend a certain amount of time on EH. The harvested energy is stored in a battery. We assume that the battery cannot charge and discharge at the same time~\cite{Luo}. In addition, the battery has sufficient charge capacity such that the amount of energy stored in the battery never reaches its maximum capacity. This assumption is reasonable since battery capacity typically ranges from joules to thousands of joules~\cite{EH_Survey}, while the energy level in the battery in our system is only in the $\mu$J range as shown in Section V.

Following the state-of-the-art EH sensor design practice~\cite{Tan}, we adopt a time-slotted or block-wise operation. We assume that one sensing operation or one transmission is performed in one time block of duration $T$ seconds.\footnote{In general a sensor may spend different amounts of time on one sensing operation~\cite{SensingPower}. Thus, the assumed protocol and analysis can be generalized to different sensing time durations other than $T$, which is outside the scope of this work.}
At the beginning of each block, we assume that the sensor checks the battery energy state and makes a decision to perform either sensing, transmission, or energy harvesting.
Thus, we define the following types of time blocks with the associated amount of energy cost/harvesting:
\begin{enumerate} [$\bullet$]
	\item
	Sensing Block (SB): the sensor samples the status information and then processes and packs sensed information into a data packet. The energy cost in a SB is denoted by $\Esen$.
	
	\item
	Transmission Block (TB): the sensor transmits the newest generated data packet (from the last sensing operation) to the sink with energy cost $\Etx$, i.e., the transmit power is $\Ptx = \Etx/T$.	
	{Then the sink sends a one-bit feedback signal to the sensor to indicate successful packet reception.  We assume that the time consumed for receiving the feedback signal at the sensor is negligible as compared to its packet transmission time.}
	If the transmission is successful, we have a successful transmission block (STB); otherwise, we have a failed transmission block (FTB). We assume that successes/failures of each TB are mutually independent~\cite{Tianqing,Ali}. \textit{The probability of a TB being a FTB, i.e., transmission outage, is denoted by $\Pout$}.
	
	\item
	Energy-harvesting block (EHB): the sensor harvests energy from the ambient environment and stores the energy in its battery.
\end{enumerate}	
\begin{figure}[t]
		
		\renewcommand{\captionlabeldelim}{ }
		\renewcommand{\captionfont}{\small} \renewcommand{\captionlabelfont}{\small}
		\centering
		\includegraphics[scale=0.8]{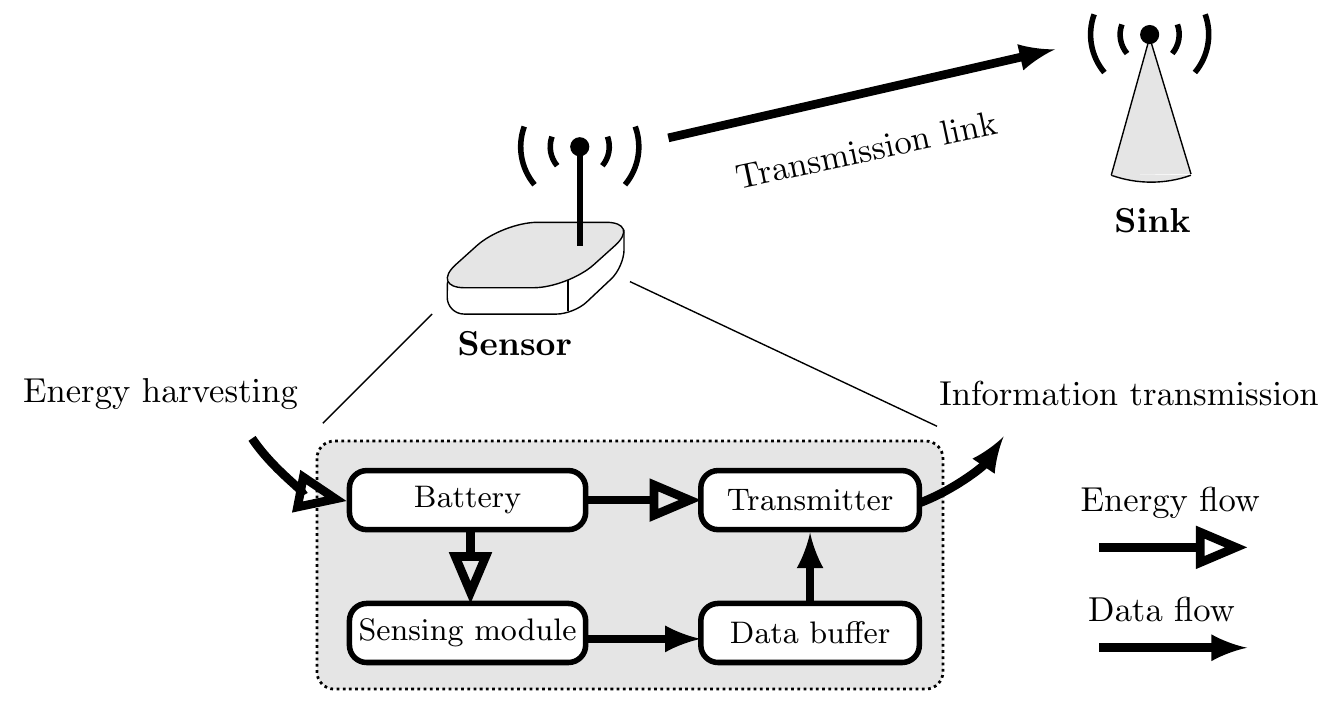}
		\vspace{-0.4cm}		
		\caption{Illustration of system model and sensor components.}
		\label{fig:0}	
		\vspace{-0.9cm}	
	\end{figure}

\subsection{Proposed Sensing and Transmission Protocol}	
Since the time-varying EH process results in randomness in the delay for performing sensing and transmission, we propose a harvest-then-use protocol with a time window for retransmissions in order to improve the delay-related performance.

The protocol is motivated as follows. \textit{Firstly}, considering the energy cost of sensing, it is necessary to harvest sufficient energy, $\Esen$, before sensing can occur. However, it is unwise to perform sensing as soon as the harvested energy reaches $\Esen$ because there will be insufficient energy left for transmission after the sensing operation. The time spent on EH due to insufficient energy for transmitting the sensed information will result in unnecessary delay. To avoid such delay, we define the condition for the sensing operation to be when the harvested energy in the battery exceeds $\Esen+\Etx$. In this way, a transmission of sensed information occurs immediately after the sensing operation (i.e., a SB is always followed by a TB). \textit{Secondly}, in the event that the transmission is not successful due to the fading channel between the sensor and sink, we need to allow for retransmissions, which are a common feature in conventional (non-EH) WSNs~\cite{retransmission_survey}. In this paper, we impose a time window for retransmissions to control the delay caused by unsuccessful transmissions because it is unwise to spend an indefinite amount of time trying to transmit outdated information. We denote $W$ as the maximum number of time blocks after a SB, within which transmissions of the currently sensed information can take place. Since the first transmission attempt always happens immediately after the SB, the time window for retransmissions is $W-1$ time blocks.
	
Under the proposed protocol, the sensor operates as follows:
\begin{enumerate}
\item 
First, the sensor uses several EHBs to harvest enough energy, $\Esen + \Etx$, and then a SB and a TB occur.

\item
If the transmission in the TB is successful, i.e., we have a STB, the sensor harvests energy (taking several EHBs) for the next sensing period {until} the battery energy exceeds $\Esen + \Etx$.

\item
If the transmission in the TB fails, i.e., we have a FTB, the sensor goes back to {harvest} energy (taking several EHBs) and performs a retransmission when the battery energy exceeds $\Etx$.

\item
Retransmission may occur several times until the sensed information is successfully transmitted to the sink or the time window for retransmissions $W-1$ is reached. Then, the data packet at the sensor is dropped and the sensor goes back to harvest the energy for a new sensing operation.
\end{enumerate}

	\begin{figure*}[!t]		
		\renewcommand{\captionlabeldelim}{ }
		\renewcommand{\captionfont}{\small} \renewcommand{\captionlabelfont}{\small}
		\centering
		\usetikzlibrary{arrows}
		\usetikzlibrary{patterns}
		\usetikzlibrary{shapes,snakes}
		\begin{tikzpicture}[scale = 0.82]

		\draw [-latex,  dash pattern=on 2pt off 3pt on 4pt off 4pt, ultra thick,gray](-5.5,0.2) -- (-5.5,1);
		\draw [-latex,  dash pattern=on 2pt off 3pt on 4pt off 4pt, ultra thick,gray](-4,0.2) -- (-4,1);
		\draw [-latex,  dash pattern=on 2pt off 3pt on 4pt off 4pt, ultra thick,gray](0.45,0.2) -- (0.45,1);
		\draw [-latex,  ultra thick,blue](-3,0.2) -- (-3,1);
		\draw [-latex,  dash pattern=on 2pt off 3pt on 4pt off 4pt, ultra thick,gray](-0.55,0.2) -- (-0.55,1);
		\draw [-latex,  -latex,  ultra thick,blue](4.45,0.2) -- (4.45,1);
		\draw [-latex,  dash pattern=on 2pt off 3pt on 4pt off 4pt, ultra thick,gray](7.45,0.2) -- (7.45,1);

		\draw [-latex,dashed,thick] (-6.7,0) -- (10.4,0);
		
		\node [rectangle,scale=1.3,fill=green!40] (v1) at (-6,0) {};
		\node [rectangle,scale=1.3,fill=none,draw,green,ultra thick] (v1) at (-5.5,0) {};
		\node [rectangle,scale=1.3,fill=none,draw,red,ultra thick] at (-5.0,0) {};
		\node [rectangle,scale=1.3,fill=none,draw,red,ultra thick] at (-4.5,0) {};
		\node [rectangle,scale=1.3,fill=none,draw,green,ultra thick] at (-4,0) {};
		\node [rectangle,scale=1.3,fill=none,draw,red,ultra thick] at (-3.5,0) {};
		\node [rectangle,scale=1.3,fill=none,draw,red,ultra thick] at (-2.5,0) {};
		\node [rectangle,scale=1.3,fill=none,draw,red,ultra thick] at (-2,0) {};
		\node [rectangle,scale=1.3,fill=green!40] (v3) at (-1.03,0) {};

		\node [rectangle,scale=1.3,fill=none,draw,green,ultra thick] at (0.45,0) {};
		\node [rectangle,scale=1.3,fill=none,draw,red,ultra thick] at (-1.5,0) {};
		\node [rectangle,scale=1.3,draw,green,ultra thick,fill=green!40] (v2) at (-3,0) {};
		\node [rectangle,scale=1.3,fill=none,draw,red,ultra thick] at (-0.05,0) {};
		\node [rectangle,scale=1.3,fill=none,draw,red,ultra thick] at (1.45,0) {};
		\node [rectangle,scale=1.3,fill=none,draw,red,ultra thick] at (2.45,0) {};
		\node [rectangle,scale=1.3,fill=none,draw,red,ultra thick] at (0.95,0) {};
		
		\node [rectangle,scale=1.3,fill=none,draw,green,ultra thick] (v3) at (-0.55,0) {};
		\node [rectangle,scale=1.3,fill=green!40] (v5) at (3.95,0) {};
		\node [rectangle,scale=1.3,draw,green,ultra thick,fill=green!40] (v5) at (4.45,0) {};
		\node [rectangle,scale=1.3,fill=none,draw,red,ultra thick] (v4) at (2.95,0) {};
		\node [rectangle,scale=1.3,fill=none,draw,red,ultra thick] (v6) at (4.95,0) {};
		\node [rectangle,scale=1.3,fill=none,draw,red,ultra thick] at (5.45,0) {};
		\node [rectangle,scale=1.3,fill=none,draw,red,ultra thick] at (5.95,0) {};
		\node [rectangle,scale=1.3,fill=none,draw,red,ultra thick] at (3.45,0) {};
		\node [rectangle,scale=1.3,fill=none,draw,red,ultra thick] at (6.45,0) {};
		
		\node [rectangle,scale=1.3,fill=green!40] at (6.95,0) {};
		\node [rectangle,scale=1.3,fill=none,draw,green,ultra thick] (v7) at (7.45,0) {};
		\node [rectangle,scale=1.3,fill=none,draw,red,ultra thick] at (7.95,0) {};

		\node [rectangle,scale=1.3,fill=none,draw,red,ultra thick] at (11.35,1.1) {};
		\node [rectangle,scale=1.3,fill=green!40] at (11.35,2.1) {};
		\node [rectangle,scale=1.3,draw,green,ultra thick,fill=green!40] at (11.35,-0.5) {};
		\node [rectangle,scale=1.3,fill=none,draw,green,ultra thick] at (11.35,-2) {};
		
		\draw [-latex,  dash pattern=on 2pt off 3pt on 4pt off 4pt, ultra thick,gray](11.35,-1.75) -- (11.35,-1.05);
		\draw [-latex,  -latex,  ultra thick,blue](11.35,-0.3) -- (11.35,0.45);

		\node [right] at (11.85,1.1) {EHB};
		\node [right] at (11.8,2.1) {SB};
		\node [right] at (11.85,-0.55) {STB};
		\node [right] at (11.85,-2) {FTB};

		\node at (-5.95,-1) {$t_{\textrm{SB},j}$};
		
		\node at (-3.15,-1) {$t_{\textrm{STB},j}$};

		\node at (4.1,-1.05) {$t_{\textrm{SB},j+1}$};
		\node at (5.7,-1.05) {$t_{\textrm{STB},j+1}$};
		
		\node at (7.15,-1) {$t_{\textrm{SB},j+2}$};
		\node at (8.85,-1) {$t_{\textrm{STB},j+2}$};

		\node at (10.3,-0.6) {t};

		\draw [decoration={brace,  raise=5,amplitude=10},decorate] (-6.9,1.1) -- (-3,1.1) node [midway,yshift=25] {};
		\draw [decoration={brace,  raise=5,amplitude=10},decorate] (-3,1.1) -- (4.45,1.1) node [midway,yshift=25] {};
		\draw [decoration={brace,  raise=5,amplitude=10},decorate] (4.45,1.1) -- (8.45,1.1) node [midway,yshift=25] {};
		\draw [decoration={brace,  raise=5,amplitude=10},decorate] (8.45,1.1) -- (10.25,1.1) node [midway,yshift=25] {};
		
		\draw [ white,opacity=1,dash pattern= on 5pt off 4pt, ultra thick, line width=5pt] (-6.7,1.45) -- (-4.55,1.45) node [midway,yshift=25] {};
		
		\draw [white,opacity=1,dash pattern= on 3pt off 4pt, ultra thick, line width=5pt] (9.4,1.45) -- (10.3,1.45);
		
		\node at (-6.5,0.5) {...};
		\node at (9.2,0.5) {...};
		
		\draw [-latex,  dash pattern=on 2pt off 3pt on 4pt off 4pt, ultra thick,gray](1.95,0.2) -- (1.95,1);
		\node [rectangle,scale=1.3,fill=none,draw,green,ultra thick] at (1.95,0) {};

		\draw [-latex,  ultra thick,blue](8.5,0.2) -- (8.5,1);
		\node [rectangle,scale=1.3,draw,green,ultra thick,,fill=green!40] at (8.5,0) {};

		\draw (-6,-0.25) -- (-6,-0.6);
		
		\draw (-3,-0.25) -- (-3,-0.6);

		\draw (3.95,-0.25) -- (3.95,-0.65);
		\draw (4.45,-0.25) -- (5.35,-0.5)  -- (5.35,-0.65);
		\draw (6.95,-0.25) -- (6.95,-0.6);
		\draw (8.5,-0.25) -- (8.5,-0.6);

		\node [circle,scale=0.2,fill=black,thick] at (-6,-0.6) {};
		
		\node [circle,scale=0.2,fill=black,thick] at (-3,-0.6) {};

		\node [circle,scale=0.2,fill=black,thick] at (3.95,-0.65) {};
		\node [circle,scale=0.2,fill=black,thick] at (5.35,-0.65) {}; 
		\node [circle,scale=0.2,fill=black,thick] at (6.95,-0.6) {};
		\node [circle,scale=0.2,fill=black,thick] at (8.5,-0.6) {};

		\draw [thick](1.7,2.45) node (v8) {} -- (-4.4,1.8);
		\draw [thick](v8) -- (6.35,1.74);

		\node [font=\large] at (-5.55,2) {...};
		\node [font=\large] at (9.7,2) {...};
		\draw [thick](v8) -- (0.75,1.8);

		\node at (1.55,2.8) {Update cycles (see Section III.B)}; %

		\draw [thick] (2.05,2.45) -- (9.3,1.7);
		
		\draw [decoration={brace,  raise=1,amplitude=10,mirror},decorate] (-6,-1.3) -- (-3,-1.3) node [midway,yshift=25] {};
		\draw [decoration={brace,  raise=1,amplitude=10,mirror},decorate] (3.95,-1.3) -- (5.45,-1.3) node [midway,yshift=25] {};
		\draw [decoration={brace,  raise=1,amplitude=10,mirror},decorate] (6.95,-1.3) -- (8.45,-1.3) node [midway,yshift=25] {};
		

		\draw [thick] (-4.45,-1.85) -- (1.4,-2.5) node (v9) {};
		\draw [thick] (v9) -- (4.65,-1.8);
		\draw [thick] (v9) -- (7.7,-1.8);

		\node at (2.2,-2.9) {Update ages (see Section III.A)}; %
		\draw [dotted,thick](10.9,2.5) -- (10.9,-2.4);
		\end{tikzpicture}
		\vspace{-0.9cm}
		\caption{Illustration of update cycle and update age.}
		\label{fig:1_1}
		\vspace*{-0.5cm}
	\end{figure*}
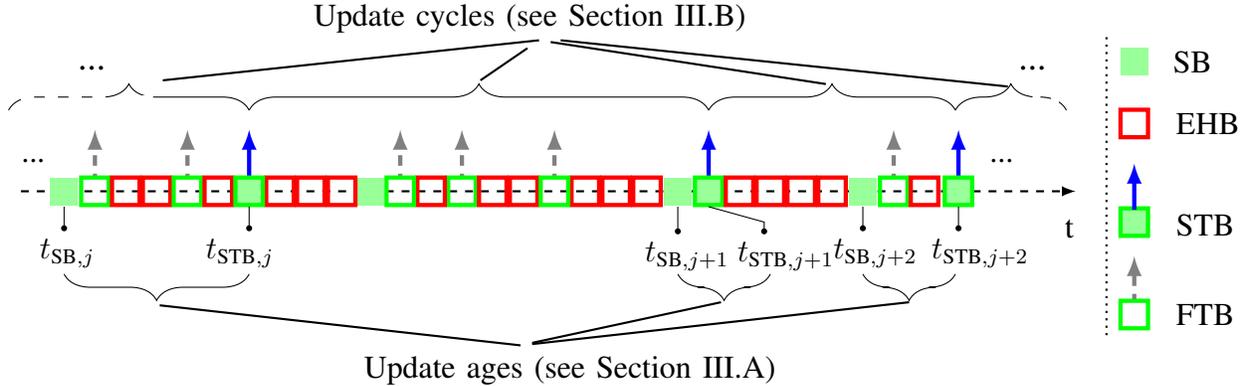	\par
Fig. \ref{fig:1_1} illustrates this protocol with $W=7$. In the example shown, the first block in Fig. \ref{fig:1_1} is a SB, followed by two FTBs (and two EHBs in between). Since the third TB is a STB, the sensed information in the first SB is successfully transmitted to the sink. Then, the sensor uses three EHBs to harvest energy to conduct sensing in the next SB. After the second SB, there are three TBs during $7$ time blocks, and all of them are FTBs. Thus, the retransmission process is terminated after $W=7$ is reached. As a result, the sensed information in the second SB is not transmitted to the sink.
The time indices shown in Fig. \ref{fig:1_1} will be defined in the following section.

\subsection{Proposed Models for Energy Arrival}\label{Sec:energyarrival}
In this paper, we consider that the harvested energy in each EHB could either \emph{remain constant} or \emph{change} from block to block. The former is referred to as deterministic energy arrival, while the latter is referred to as random energy arrival.

Deterministic energy arrival is an appropriate model when the coherence time of the EH process is much larger than the duration of the entire communication session, such as EH by solar panel on clear days~\cite{Yates,eu2011design,chuan_huang_relay}. In this paper, we denote this as \emph{deterministic energy arrival process}. For tractability, we also assume that $\Esen$ and $\Etx$ represent integer multiples of the harvested energy by one EHB, $\rho$.

For random energy arrivals, we consider independent and identically distributed (i.i.d.) random energy arrival model\footnote{The i.i.d. energy arrival model is commonly considered in the literature~\cite{Luo,Morsi,Yishun}. There are other energy arrival models captures the temporal correlation of the energy arrival process, such as discrete-Markovian modeling~\cite{Niyato,Jing_lei,Ho_markov}, which are beyond the scope of this work.} with a general probability distribution function for the amount of energy harvested in each EHB. This energy arrival model is referred to as \emph{general random energy arrival process}. The previously considered exponential and gamma distributions in~\cite{Tianqing,Ali,Luo} become as special cases of the general probability distribution in this work. Since the exponential distribution is commonly studied for wireless power transfer using RF signals, we will also provide results for this important special case and referred to it as \emph{exponential energy arrival process}.

\section{Delay-Related Metrics}
As described in the previous section, both sensing and (re)transmission	requires a variable amount of EH time, which may result in significant delays in obtaining the sensed information at the sink. In this section, we consider two metrics to measure the delay performance of the considered sensing and transmission protocol.

For the convenience of describing the two metrics, as shown in Fig. \ref{fig:1_1}, we use $\teventsuc{j}$ to denote the block index for the $j$th STB during the entire sensing and transmission operation. Note that a successful transmission also induces an information update at the sink. Also, it is important to associate each transmission with its information content. To this end, we use $\teventsen{j}$ to denote the block index for the SB in which the sensed information is transmitted in the $j$th STB. In other words, status information sensed at $\teventsen{j}$ is successfully transmitted to the sink at $\teventsuc{j}$. Next, we define two delay-related metrics, expressed in terms of the number of time blocks:

\subsection{Update Age and Update Cycle}
{
\begin{definition}[Update age]
	\textnormal{For the $j$th STB, the update age is given by the number of time blocks from $\teventsen{j}$ to $\teventsuc{j}$ (shown in Fig.~\ref{fig:1_1}). The $j$th update age is}
	\begin{equation} \label{eq:def_Tstd}
	{\Tstd}_{,j}
	= \teventsuc{j}-\teventsen{j},\ j = 1,2,3,....
	\end{equation}
\end{definition}

\begin{remark}
	\textnormal{
The update age measures the time elapsed from the generation of a status-information-containing packet at the sensor to the reception of the packet, i.e., status update, at the sink. This metric is referred to as the status update age in~\cite{R1}.
A larger update age implies that a more outdated status is received by the sink.
The update age, which captures the freshness of the updated status information, however, does not reflect the update frequency at the sink.
Rather, the update frequency is captured by the update cycle which is presented below:}
\end{remark}

\begin{definition}[Update cycle]
\textnormal{For the $j$th STB, the update cycle is given by the number of time blocks from $\teventsuc{j-1}$ to $\teventsuc{j}$ (shown in Fig.~\ref{fig:1_1}). The $j$th update cycle is}		
\begin{equation} \label{eq:def_Tuc}
\hspace*{0.35cm}{\Tuc}_{,j}
= \teventsuc{j+1}-\teventsuc{j},\ j = 1,2,3,....
\end{equation}
\end{definition}

\begin{remark}	
	\textnormal{
				The update cycle measures the time elapsed from one status update at the sink to the next.
				The update cycle, however, does not reflect the update freshness at the sink.
				Unlike the update age, the update cycle takes into account the delay due to dropped data packets.
				Therefore, update cycle complements update age, and they jointly capture the update frequency and freshness, 
				to provide comprehensive metrics on the delay performance of a status monitoring system.
	}
\end{remark}}

\subsection{Modeling Delay-Related Metrics as i.i.d. Random Variables}

To model each of the update age/update cycle as i.i.d. random variables, we focus on the steady-state behavior as characterized in Lemma 1.
%
\begin{lemma} \label{L1} 
	\textnormal{For a deterministic energy arrival process, the energy level after each TB is zero.
		For a general random energy arrival process with pdf {containing at least one positive right-continuous point}, $f(\epsilon)$, the steady-state distribution of the energy level after each TB has pdf}
		\begin{equation}
		g\left(\epsilon\right) = \frac{1}{\rho} \left(1- F(\epsilon)\right),
		\end{equation}
	\textnormal{where $\rho$ is the average harvested energy, and $F(\epsilon)$ is the cumulative distribution function (cdf) corresponding to $f(\epsilon)$.}
\end{lemma}
\begin{IEEEproof}
			For a deterministic energy arrival process, Lemma 1 is straightforward. For a general random energy arrival process, the proof is given in Appendix B.
\end{IEEEproof}

According to the sensing and transmission protocol defined in the previous section, each SB is directly followed by a TB. From Lemma 1, the steady-state distribution of available energy after any TB is the same. Hence, the steady-state distribution of the available energy after  $\teventsuc{j}$ is the same for all $j$.
Because the successes/failures of each TB are mutually independent, and ${\Tuc}_{,j}$ is determined by both the available energy after $\teventsuc{j}$ and the successes/failures of the following TBs, ${\Tuc}_{,j}$ are i.i.d. for all $j$.
Similarly, it is also easy to show that ${\Tstd}_{,j}$ are i.i.d. for all $j$.
For convenience, we remove subscript $j$ for ${\Tuc}$ and $\Tstd$ in \eqref{eq:def_Tuc} and \eqref{eq:def_Tstd}, respectively.

\section{Update Age}
In this section, considering the dynamics of an energy arrival process and the probability of successful/failed transmission in our proposed harvest-then-use protocol, the update age for deterministic, general random and exponential energy arrival processes are analyzed.

\subsection{Deterministic Energy Arrival Process}
\begin{theorem}
	For a deterministic energy arrival process, the update age pmf is given by
	\begin{equation} \label{det_STD_pmf}
	\begin{aligned}
	\myprobability{\Tstd = k} =
	&\frac{(1-\Pout) \Poutk{n-1}}{\Psuc}, & k=1+ (n-1)\left(\frac{\Etx}{\rho} +1\right),
	\end{aligned}
	\end{equation}
	where
	\begin{equation} \label{just_hat_n}
	n = 1,2,... \hat{n},\ \hat{n} = 1+ \left\lfloor \frac{W-1}{1+ \frac{\Etx}{\rho}}  \right\rfloor,
	\Psuc=1-\Poutk{\hat{n}},
	\end{equation}
	and $\Pout$ is the probability of a TB being a FTB, defined in Section II.
\end{theorem}

\begin{proof}
	See Appendix D.
\end{proof}
\noindent From Theorem 1, the average update age, $\aveTstd$ for a deterministic energy arrival process is straightforwardly obtained as in Corollary 1.
\begin{corollary}
	For a deterministic energy arrival process, average update age is given by
	\begin{equation} \label{det_STD_ave}
	\begin{aligned}
	\aveTstd=
	\sum\limits_{n=1}^{\hat{n}}
	\left(1+(n-1)\left(\frac{\Etx}{\rho} +1\right)\right)
	\frac{(1-\Pout) \Poutk{n-1}}{\Psuc},
	\end{aligned}
	\end{equation}
	where $\Psuc$ is given in \eqref{just_hat_n}.
\end{corollary}

\subsection{General Random Energy Arrival Process}
\begin{theorem}
	For a general random energy arrival process, the update age pmf \!\footnote{{Although the general expression in Theorem 2 contains multiple integrals in Eq. (9), for special cases, such as deterministic and exponential energy arrival processeses, the results given in Theorems 1 and 3 are closed-form expressions.}} is given by
	\begin{equation} \label{gen_STD_pmf}
	\myprobability{\Tstd \!\!= \!\!k}\!\! = \!\!
	\left\lbrace
	\begin{aligned}
	&\!\!\frac{1-\Pout}{\Psuc}, \!\!&k=1, \\
	&\!\!\frac{(1-\Pout)}{\Psuc} \!\!\sum\limits_{n=2}^{k} \Poutk{n-1}\!\! \left(G_{k-n-1}((n-1)\Etx)\!-\!G_{k-n}((n-1)\Etx)\right),\!\!\!\! &2\!\leq \!k\! \leq\! W,\\
	\end{aligned}
	\right.
	\end{equation}
	where
	\begin{equation} \label{Psuc_gen}
	\begin{aligned}
	\Psuc =
	1-\Pout + (1-\Pout) \sum\limits_{l=2}^{W} \sum\limits_{n=2}^{l} \Poutk{n-1}
	\left(G_{l-n-1}((n-1)\Etx)-G_{l-n}((n-1)\Etx)\right),
	\end{aligned}
	\end{equation}
	and
	\begin{equation} \label{defi_G}
	G_{i}(x) =
	\left\lbrace
	\begin{aligned}
	&1, &i=-1,\\
	&\int\limits_{0}^{x} (g\underbrace{\star f \star f \star...\star f}_{i \text{ convolutions}})(u) \mathrm{d}u, &i\geq 0.
	\end{aligned}
	\right.
	\end{equation}
	$g(x)$ and $f(x)$ are defined in Lemma 1.
\end{theorem}
\begin{proof}
	See Appendix D.
\end{proof}

From Theorem 2, the average update age, $\aveTstd$ for a general random energy arrival process is obtained straightforwardly as in Corollary~2.
\begin{corollary}
	For a general random energy arrival process, average update age is given by
	\begin{equation}\label{gen_STD_ave}
	\aveTstd
	= \frac{1-\Pout}{\Psuc} \left(
	1+
	\sum\limits_{l=2}^{W} l
	\sum\limits_{n=2}^{l} \Poutk{n-1} \left(G_{l-n-1}((n-1)\Etx)-G_{l-n}((n-1)\Etx)\right)
	\right),
	\end{equation}
	where $\Psuc$ is given in \eqref{Psuc_gen}.
\end{corollary}

\subsection{Exponential Energy Arrival Process}
\begin{theorem}
	For an exponential energy arrival process, the update age pmf is given by
	\begin{equation} \label{Tstd_exp}
	\myprobability{\Tstd = k}
	=\left\lbrace
	\begin{aligned}
	& \frac{1-\Pout}{\Psuc}, k =1,\\
	&\frac{(1-\Pout)}{\Psuc} \sum\limits_{n=2}^{k} \Poutk{n-1}
	\Pois{k-n}{(n-1)\Etx/\rho},
	2 \leq k \leq W,
	\end{aligned}
	\right.
	\end{equation}
	where
	\begin{equation} \label{Psuc_exp}
	\Psuc
	= 1-\Pout + (1-\Pout) \sum\limits_{l=2}^{W} \sum\limits_{n=2}^{l} \Poutk{n-1}
	\Pois{l-n}{(n-1)\Etx/\rho},
	\end{equation}
\end{theorem}

\begin{proof}
	See Appendix D.
\end{proof}
\noindent From Theorem 3, the average update age, $\aveTstd$ for an exponential energy arrival process is straightforwardly obtained as in Corollary 3.

\begin{corollary}
	For an exponential energy arrival process, average update age is given by	
	\begin{equation} \label{exp_STD_ave}
	\begin{aligned}
	\aveTstd=
	\frac{(1-\Pout)}{\Psuc}
	\left(1+
	\sum\limits_{l=2}^{W}
	l
	\sum\limits_{n=2}^{l} \Poutk{n-1}
	\Pois{l-n}{(n-1)\Etx/\rho}
	\right)
	\end{aligned}
	\end{equation}
	where $\Psuc$ is given in \eqref{Psuc_exp}.
	
\end{corollary}

From Theorems 1-3 and Corollaries 1-3, we see that different energy arrival processes induce different pmfs and average values of update age.
For benchmarking with the existing studies on delay without imposing a constraint on the time window for retransmissions~\cite{Tianqing}, we let $W\rightarrow \infty$, so that all sensed information is eventually transmitted to the sink, the average update age is the same under different energy arrival processes as in Corollary~4.

%
\begin{corollary}
	For a deterministic or general random energy arrival process, $\aveTstd$ increases with $W$, and as $W$ gets large, the asymptotic upper bound of $\aveTstd$ is independent with energy arrival distribution and is given by
	\begin{equation}
	\lim\limits_{W\rightarrow \infty} \aveTstd = 1+ \frac{\Pout}{1-\Pout} \left(\frac{\Etx}{\rho} +1\right).
	\end{equation}
\end{corollary}
\begin{proof}
	See Appendix G.
\end{proof}

\begin{remark} \label{R1}
	\textnormal{
		From the above analytical results, we have that:
		\begin{enumerate}[i)]
			\item
			From Theorems 1 to 3, $\Tstd$ is independent of the energy cost of sensing, $\Esen$, because the delay is only affected by the energy harvesting and retransmissions that happen after the sensing operation.
			This might give the impression that energy cost of sensing does not affect delay. However, update age is only one of the two delay metrics, and the energy cost of sensing has important impacts on update cycle, which will be investigated in the next section.		
			\item
			Allowing a larger window for retransmissions increases the average update age. This might suggest that retransmissions should be avoided, i.e., $W=1$.
			However, the update age does not take into account cases where sensed information is not successfully transmitted to the sink. In this regard, the update cycle implicitly captures such cases.
		\end{enumerate}
	}
\end{remark}

\section{Update Cycle}
In this section, considering the dynamics of an energy arrival process and the probability of successful/failed transmission in our proposed harvest-then-use protocol, the update cycle for deterministic, general random and exponential energy arrival processes are analyzed.

\subsection{Deterministic Energy Arrival Process}
\begin{theorem}
	For a deterministic energy arrival process, the update cycle pmf is given by
	\begin{equation} \label{det_UC_pmf}
	\begin{aligned}
	\myprobability{\Tuc\! =\! k} &\!=\!
	(1-\Pout) \Poutk{n-1+m\hat{n}}, \
	\!\!k \!=\! \left(\!\!\frac{\Esen + n \Etx}{\rho}\!\!\right) \!+ \!n\!+\!1\!+\! m \!\left(\!\frac{\Esen + \hat{n}\Etx}{\rho} \!+\! (\hat{n}\!+\!1)\!\right)\!\!,
	\end{aligned}
	\end{equation}
	where $n = 1,2,... \hat{n}, m = 0,1,2,...$, and $\hat{n}$ is given in \eqref{just_hat_n}.
\end{theorem}

\begin{proof}
	See Appendix E.
\end{proof}

\begin{corollary}
	For a deterministic energy arrival process, average of update cycle is given by	
	\begin{equation} \label{det_UC_ave}
	\aveTuc
	= \frac{\Poutk{\hat{n}} }{1-\Poutk{\hat{n}}}
	\left(1+\hat{n} + \frac{\Esen \!+\! \hat{n} \Etx}{\rho}\right) + \frac{\Esen}{\rho} +1 \!+\! (1+\frac{\Etx}{\rho}) \frac{1-\Pout}{1-\Poutk{\hat{n}}}
	\sum\limits_{n=1}^{\hat{n}} \Poutk{n-1} n.
	\end{equation}
\end{corollary}

\begin{proof}
	See Appendix F.
\end{proof}

\subsection{General Random Energy Arrival Process}
\begin{theorem}
	For a general random energy arrival process, the update cycle pmf is given by
	\begin{equation} \label{gen_UC_pmf}
	\begin{aligned}
	\myprobability{\Tuc \!=\! k} \!=\!\!
	\sum\limits_{m=0}^{\hat{m}}\!
	\!\left(\!
	\zeta(\mathcal{\Esen\!+\!\Etx})
	\underbrace{\ast\zeta(\mathcal{{\Esen}})\ast \cdots \ast\zeta(\mathcal{{\Esen}})}_{m \text{ convolutions}}
	\underbrace{\ast\vartheta\ast \cdots \ast\vartheta}_{m \text{ convolutions}} \ast \iota
	\!\right)
	\!\!(k\!-\!m(1\!+\!W)\!-\!1), k=2,3,....
	\end{aligned}
	\end{equation}
\noindent
where $\hat{m}=\left\lfloor \frac{k-2}{W+1} \right\rfloor$, and functions $\zeta(\mathcal{E},i)$, $\iota(i)$ and $\vartheta(i)$ are given by
\begin{subequations}
	\begin{align}
	&\zeta(\mathcal{E},i) = G_{i-1}(\mathcal{E}) -G_{i}(\mathcal{E}),\\
	&\iota(i) = \Psuc \myprobability{\Tstd =i},\\
	&\begin{aligned}
	&\vartheta(i) = \Pout \left(G_{W+i-2}(\Etx)-G_{W+i-1}(\Etx)\right)  \\
	&+
	\!\sum\limits_{l=2}^{W} \! \sum\limits_{n=2}^{l} \!\Poutk{n}
	\left(G_{l-n-1}((n-1)\Etx)\!-\!G_{l-n}((n-1)\Etx)\right)
	\left(G_{W+i-l-1}(\Etx)\!-\!G_{W+i-l}(\Etx)\right).
	\end{aligned}
	\end{align}
\end{subequations}
$\myprobability{\Tstd=i}$, $\Psuc$ and $G_i(\mathcal{E})$ are given in
\eqref{gen_STD_pmf}, \eqref{Psuc_gen} and \eqref{defi_G}, respectively.
\end{theorem}
\begin{proof}
	See Appendix E.
\end{proof}

\begin{corollary}
	For a general random energy arrival process, average update cycle is given by
	\begin{equation} \label{gen_UC_ave}
	\aveTuc
	={\frac{1-\Psuc}{\Psuc}} \left({\frac{\Esen}{\rho} + \bar{V} + W + 1}\right) + \frac{\Esen + \Etx}{\rho} + \aveTstd +1,
	\end{equation}
	where $\Psuc$ and $\aveTstd$ are respectively given in \eqref{Psuc_gen} and \eqref{gen_STD_ave}, and
	\begin{equation}
	\begin{aligned}
	\bar{V}
	&=\frac{\Pout}{1-\Psuc} \left(\frac{\Etx}{\rho} - \sum\limits_{i=0}^{W-2}i \left(G_{i-1}(\Etx)-G_{i}(\Etx)\right) - \!(W\!-\!1) \!\left({\!1\!-}\!\!\sum\limits_{i=0}^{W-2}\!\left(G_{i-1}(\Etx)-G_{i}(\Etx)\right)\! \right)\!\right)\\
	&+\frac{1}{1-\Psuc}
	\sum\limits_{l=2}^{W} \sum\limits_{n=2}^{l}
	\left(\Pout\right)^n \left(G_{l-n-1}((n-1)\Etx)-G_{l-n}((n-1)\Etx)\right)\times\\
	&\left(\frac{\Etx}{\rho} - \sum\limits_{i=0}^{W-l-1}i\left(G_{i-1}(\Etx)-G_{i}(\Etx)\right) - (W-l) \left({1-}\sum\limits_{i=0}^{W-l-1}\left(G_{i-1}(\Etx)-G_{i}(\Etx)\right) \right)\right).
	\end{aligned}
	\end{equation}
\end{corollary}

\begin{proof}
	See Appendix F.
\end{proof}

\subsection{Exponential Energy Arrival Process}
\begin{theorem}
	For an exponential energy arrival process, the update cycle pmf is given by
	\begin{equation} \label{exp_UC_pmf}
	\begin{aligned}
	\myprobability{\Tuc = k} =
	\sum\limits_{m=0}^{\hat{m}}
	\left(
	\zeta((m+1) {\Esen}+ {\Etx})
	\underbrace{\ast\vartheta\ast \cdots \ast\vartheta}_{m \text{ convolutions}} \ast \iota
	\right)
	\!\!(k-m(1+W)-1),\ k=2,3,....
	\end{aligned}
	\end{equation}
where $\hat{m}=\left\lfloor \frac{k-2}{W+1} \right\rfloor$, and functions $\zeta(\mathcal{E},i)$, $\iota(i)$ and $\vartheta(i)$ are given by
\begin{subequations}
	\begin{align}
	&\zeta(\mathcal{E},i) = \Pois{i}{\mathcal{E}/\rho},\\
	&\iota(i) = \Psuc \myprobability{\Tstd =i},\\
	&\begin{aligned}
	&\vartheta(i) = \Pout \Pois{W+i-1}{\Etx / \rho} \\
	&\hspace{0.8cm}+
	\sum\limits_{l=2}^{W}  \sum\limits_{n = 2}^{l} \Poutk{n}
	\Pois{l-n}{(n-1)\Etx/\rho}
	\Pois{W+i-l}{\Etx /\rho},
	\end{aligned}
	\end{align}
\end{subequations}
and $\myprobability{\Tstd=i}$ and $\Psuc$ are given in
\eqref{Tstd_exp} and \eqref{Psuc_exp}, respectively.
\end{theorem}
\begin{proof}
	See Appendix E.
\end{proof}

\begin{corollary}
	For an exponential energy arrival process, average update cycle is given by
	\begin{equation} \label{exp_UC_ave}
	\aveTuc
	={\frac{1-\Psuc}{\Psuc}} \left({\frac{\Esen}{\rho} + \bar{V} + W + 1}\right) + \frac{\Esen + \Etx}{\rho} + \aveTstd +1,
	\end{equation}
	where $\aveTstd$ and $\Psuc$ are given in \eqref{exp_STD_ave} and \eqref{Psuc_exp}, and
	\begin{equation}
	\begin{aligned}
	&\bar{V}
	= \frac{{\Pout}}{{1-\Psuc}}
	\left( {\frac{\Etx}{\rho}} - \sum\limits_{i=0}^{W-2} i \Pois{i}{\frac{\Etx}{\rho}} - (W-1) \left(1-\sum\limits_{i=0}^{W-2} \Pois{i}{\frac{\Etx}{\rho}}\right) \right)
	\!+\! \frac{1}{{1\!-\!\Psuc}} \times \\
	&\sum\limits_{l=2}^{W} \sum\limits_{n = 2}^{l}
	\Poutk{n} {\Pois{l\!-\!n}{(n\!-\!1)\frac{\Etx}{\rho}}} \!\!
	\left( \!{\frac{\Etx}{\rho}} \!-\!\!\!\!\! \sum\limits_{i=0}^{W\!-\!l\!-\!1}\!\! i \Pois{i}{\frac{\Etx}{\rho}} \!-\! (W\!-\!l) \!\left(\!\!1\!-\!\!\sum\limits_{i=0}^{W\!-\!l\!-\!1} \!\!\Pois{i}{\frac{\Etx}{\rho}}\!\right)\! \!\right)\!.
	\end{aligned}
	\end{equation}
\end{corollary}

\begin{proof}
	See Appendix F.
\end{proof}

	Similar with the case of update age,
	different energy arrival processes induces different pmfs and average values of update cycle.
	However,
	for benchmarking with the existing studies on delay without imposing a constraint on the maximum allowable retransmission time,
	when we consider removing the constraint of retransmission, i.e., $W\rightarrow \infty$, so that all sensed information is eventually transmitted to the sink, the average update cycle is the same under different energy arrival processes as in Corollary~8.

\begin{corollary}	
	For a deterministic or general random energy arrival process, $\aveTuc$ decreases with $W$, and as $W$ grows large, the asymptotic lower bound of $\aveTuc$ is independent with energy arrival distribution and is given by
	\begin{equation}
	\lim\limits_{W\rightarrow \infty} \aveTuc = 2+ \frac{\Esen + \Etx}{\rho} + \frac{\Pout}{1-\Pout} \left(\frac{\Etx}{\rho} +1\right).
	\end{equation}
\end{corollary}
\begin{proof}
	See Appendix G.
\end{proof}

\begin{remark} \label{R2}
	\textnormal{
		From the above analytical results, we have that:
		\begin{enumerate}[i)]
			\item
			From Theorems 4 to 6, we know that ${\Tuc}$ is affected by the energy cost of sensing, $\Esen$. A larger $\Esen$ means more EHBs are required to harvest a sufficient amount of energy to perform sensing operation(s) between adjacent STBs.
			\item
			A larger window for retransmissions shorten the average update cycle,
			because allowing more retransmissions increases the chance of having a successful transmission. 
			This might suggest that it is also better to increase $W$ to reduce the update cycle. But increasing $W$ also increases the update age as discussed earlier. Therefore, there is clearly a tradeoff between the two metrics.
		\end{enumerate}		}
	\end{remark}

\section{Numerical Results}
In this section, we present numerical results for the update age and update cycle, using the results in Theorems 1-6 and Corollaries 1-8. 
The typical outdoor range for a wireless sensor is from $75$~m to $100$~m \cite{Micaz}. Hence, we set the distance between the sensor and the sink as $d=90$~m and the path loss exponent for the sensor-sink transmission link as $\lambda= 3$~\cite{Tianqing}. The duration of a time block is $T = 5$~ms\cite{eu2011design}. The noise power at the sink is $\sigma^2 = -100$ dBm~\cite{Ali}. The average harvested power is $10$~mW \cite{Fan_Zhang}, i.e., average harvested energy per time block, $\rho = 50\ \mu$J. Unless otherwise stated, (i) we set the power consumption in each TB, $\Ptx = 40$~mW, i.e., $\Etx = 200\ \mu$J. Note that this includes RF circuit consumption (main consumption) and the actual RF transmit power $\Prf = -5$~dBm\footnote{The values we chose for $\Ptx$ and $\Prf$ are typical for commercial sensor platforms, such as MICAz~\cite{Micaz}.} and (ii) we set the power consumption in each SB as $\Psen = 50$~mW\cite{SensingPower}, i.e., $\Esen = 250\ \mu$J.
{In the following calculations, power and SNR related quantities use a linear scale.}
We assume that a transmission outage from the sensor to the sink occurs when the SNR at the sink $\gamma$, is lower than SNR threshold $\gamma_0 =40$~dB~\cite{Ding}. The outage probability is
\begin{equation}
\Pout = \myprobability{\gamma< \gamma_0}.
\end{equation}
The SNR at the sink is~{\cite{Goldsmith}}
\begin{equation} \label{snr}
\gamma = \frac{\vert h \vert^2 \Prf}{{\Gamma} d^{\lambda} \sigma^2},
\end{equation}
where $h$ is the source-sink channel fading gain,
{$\Gamma = \frac{P_L(d_0)}{d^{\lambda}_0}$, is a path loss factor relative to reference distance $d_0$ of the antenna far
field, and $P_L(d_0)$ is linear-scale path loss, 
which depends on the propagation environment~\cite{Tianqing}.
Following~\cite{Tianqing,Ali}, we assume $\Gamma =1$, for simplicity.}

For the numerical results, we assume that $h$ is block-wise Rayleigh fading.
Using \eqref{snr}, the outage probability can be written as
\begin{equation} \label{real_Pout}
\Pout = 1- \exp\left(- \frac{ d^{\lambda} \sigma^2 \gamma_0}{\Prf}\right).
\end{equation}
By applying \eqref{real_Pout} to the theorems and corollaries in Sections IV and V, we compute the expressions of the pmfs of $\Tstd$ and $\Tuc$ as well as their average values $\aveTstd$ and $\aveTuc$.

\textbf{Pmfs of update age and update cycle with different energy arrival processes:} First, we consider a deterministic energy arrival process with harvested energy in each EHB, $\rho$. Also we consider two special cases of the general random energy arrival process: (i) exponential energy arrival processes with average harvested energy in each EHB, $\rho$ and (ii) random energy arrival processes with gamma distribution~\cite{Luo}, $\mathrm{Gamma}(0.05,1000)$. We term this as the \emph{gamma energy arrival process}, and it is easy to verify that this gamma energy arrival process has the same average harvested energy in each EHB as the deterministic and exponential energy arrival processes.

Figs.~\ref{fig:PMF_det}-\ref{fig:PMF_exp} plot the pmfs of update age, $\Tstd$, and update cycle, $\Tuc$, for the deterministic, gamma and exponential energy arrival process, respectively. The analytical results are plotted using Theorems 1-6,
and we set $W=50$, i.e., the time window for retransmissions is $W-1=49$ time blocks.
In particular, in Fig.~\ref{fig:PMF_gamma} the analytical pmfs of $\Tstd$ and $\Tuc$ for the general random arrival process are obtained using Theorems 2 and 5.
The results in Figs.~\ref{fig:PMF_gamma}-\ref{fig:PMF_exp} also illustrate the importance of the general random energy arrival process, which is used in this work. This is because gamma and exponential energy arrival processes, which have been used in the literature~\cite{Tianqing,Ali,Ding,Luo}, are special cases of the general random energy arrival process.
{\textit{We see that different energy arrival processes result in different pmfs of update age and update cycle.
		Hence, a statistical analysis of the two metrics will provide insight into the design of future~EH~WSNs.}}

In the following figures (Figs. 6-9), we only present the numerical results for the average values of the two delay metrics, which have been presented in Corollaries 1-8.
\begin{figure*}[t]
	\renewcommand{\captionfont}{\small} \renewcommand{\captionlabelfont}{\small}
	\minipage{0.35\textwidth}
	\includegraphics[width=\linewidth]{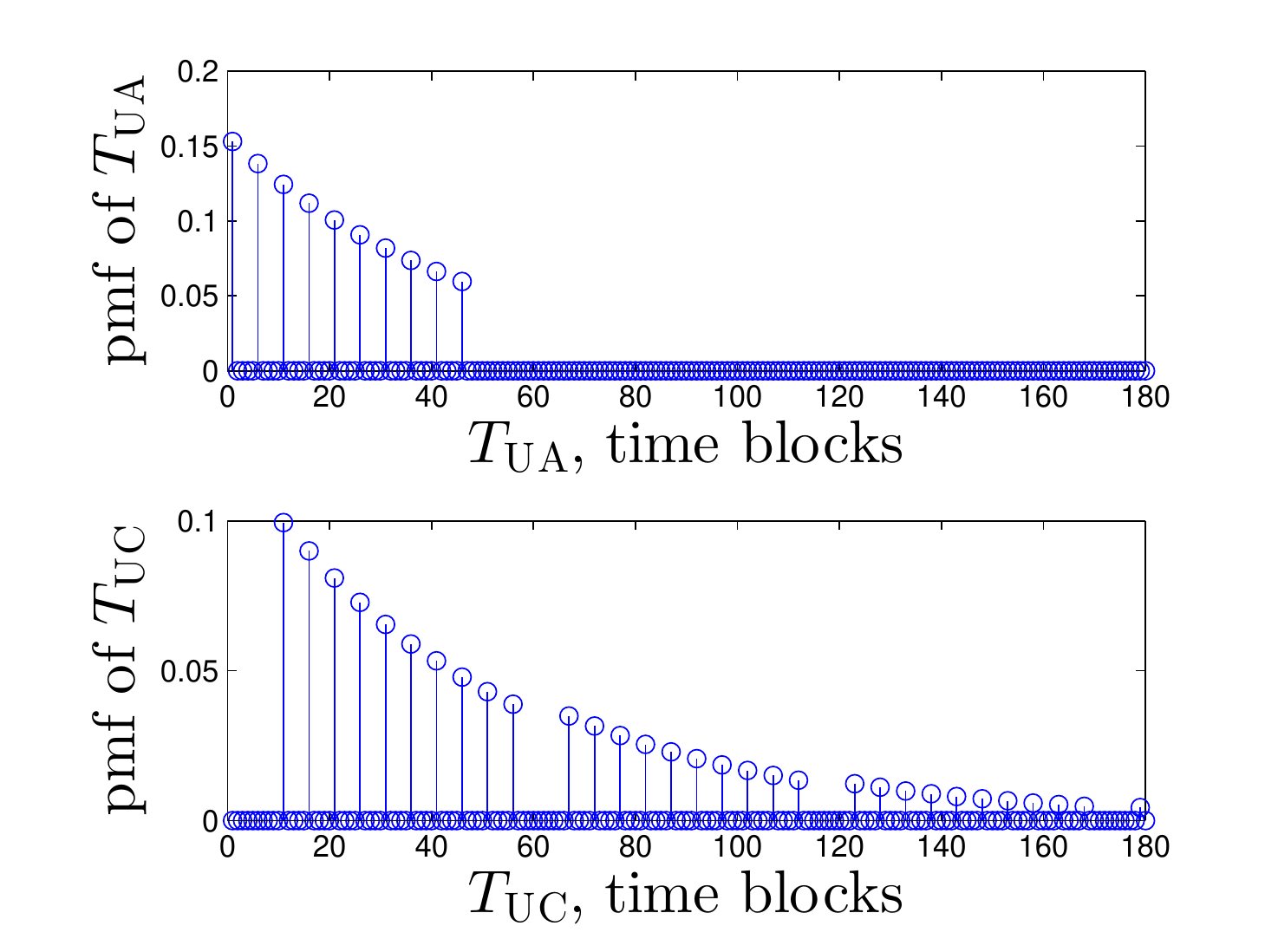}
	\vspace*{-1.3cm}
	\caption{pmfs\! for\! $\Tstd$\! and\! $\Tuc$\! with\! \newline deterministic\! energy\! arrival\! process.}\label{fig:PMF_det}
	\endminipage
		\hspace{-15pt}
	\minipage{0.35\textwidth}
	\includegraphics[width=\linewidth]{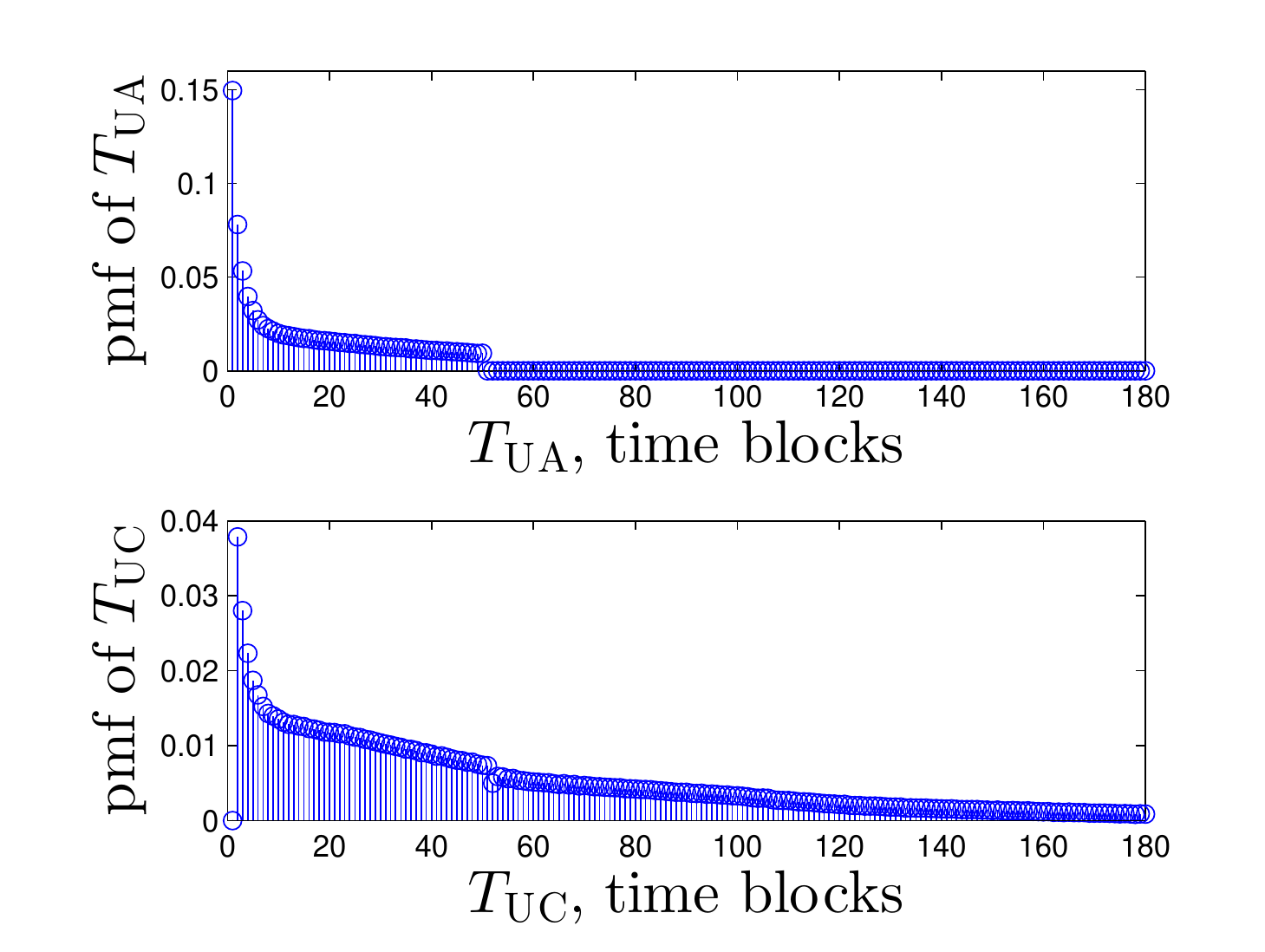}
	\vspace*{-1.3cm}
	\caption{pmfs\! for\! $\Tstd$\! and\! $\Tuc$\! with\! \newline gamma\! energy\! arrival\! process.}\label{fig:PMF_gamma}
	\endminipage
	\hspace{-15pt}	
	\minipage{0.35\textwidth}
	\includegraphics[width=\linewidth]{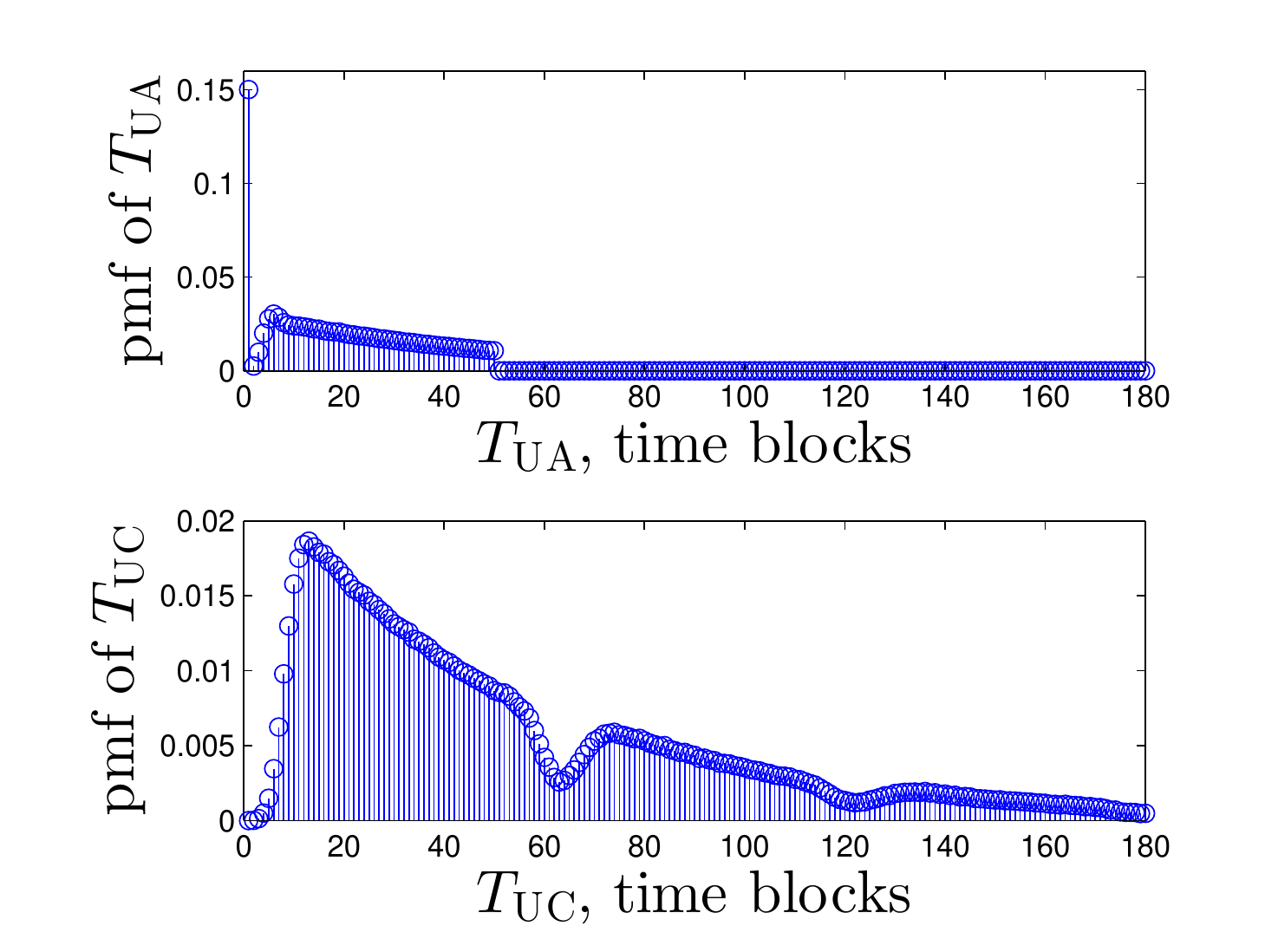}
	\vspace*{-1.3cm}
	\caption{pmfs\! for\! $\Tstd$\! and\! $\Tuc$\! with\! \newline exponential\! energy\! arrival\! process.}\label{fig:PMF_exp}
	\endminipage\\
	\hrulefill
	\vspace*{-0.7cm}
\end{figure*}

\textbf{Average update age and average update cycle with different energy arrival processes:} Figs. \ref{fig:diff_distribution_STD} and \ref{fig:diff_distribution_UC} show the average update age, $\aveTstd$, and the average update cycle, $\aveTuc$, for different $W$, i.e., different time windows for retransmissions, $W-1$, and energy arrival processes. The results in Figs.~\ref{fig:diff_distribution_STD} and \ref{fig:diff_distribution_UC} are generated using Corollaries 1-4 and Corollaries 5-8, respectively. We can see that the different energy arrival models result in almost the same values of the average update age and especially the average update cycle. As the time window for retransmissions increases, the average update age increases monotonically and approaches its analytical upper bound given by Corollary 4, while the average update cycle decreases monotonically and approaches its analytical lower bound given by Corollary 8.
{\textit{Thus, with a smaller time window for retransmissions, the updated status is more fresh, but the update frequency is lower.}}
\begin{figure*}[t]
	\renewcommand{\captionfont}{\small} \renewcommand{\captionlabelfont}{\small}
	\minipage{0.49\textwidth}
	\includegraphics[width=\linewidth]{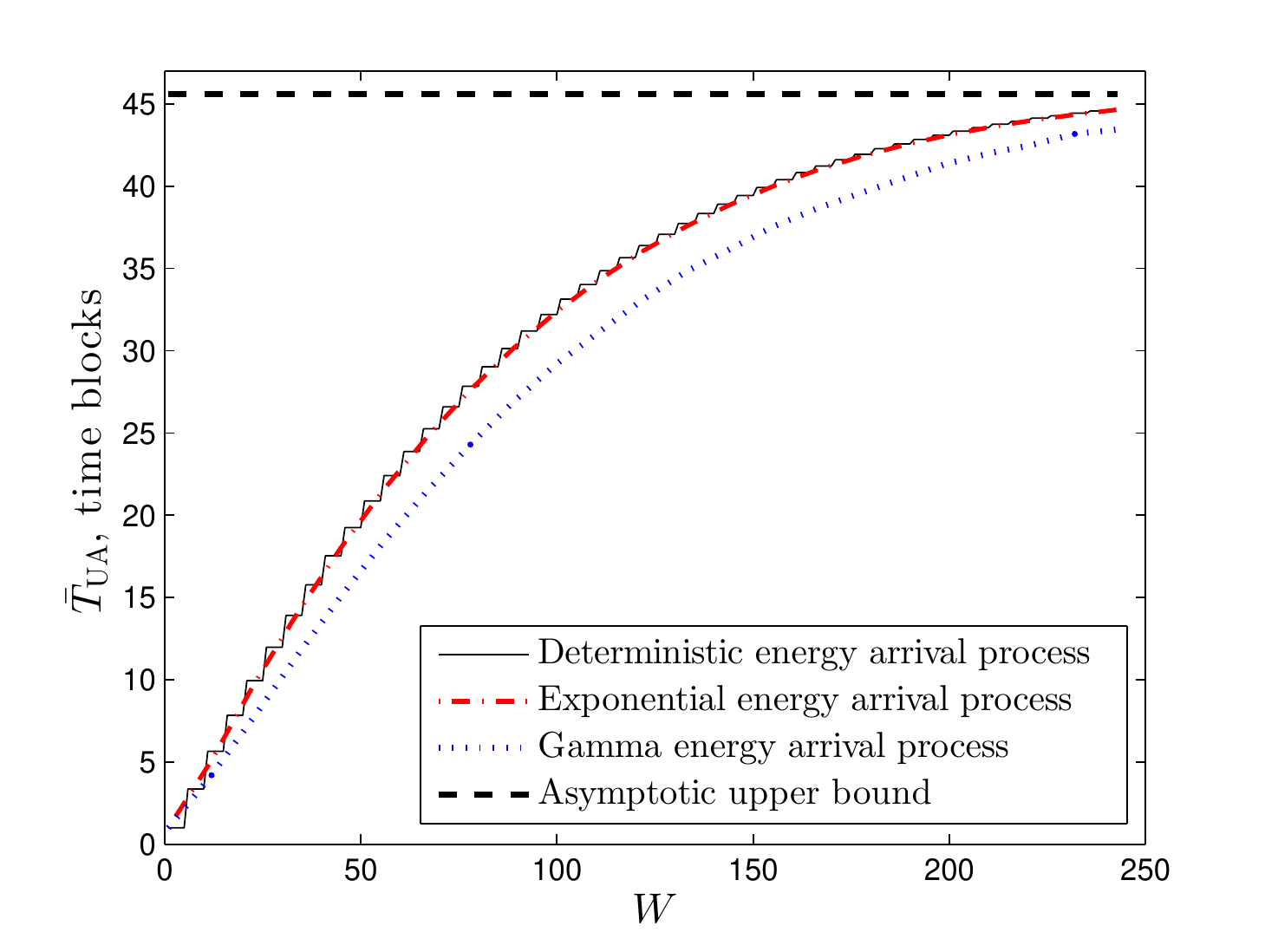}
	\vspace*{-1.3cm}
	\caption{Average update age, $\aveTstd$, versus $W$, with different energy arrival processes.}\label{fig:diff_distribution_STD}
	\endminipage\hfill
	\hspace{-5pt}
	\minipage{0.49\textwidth}
	\includegraphics[width=\linewidth]{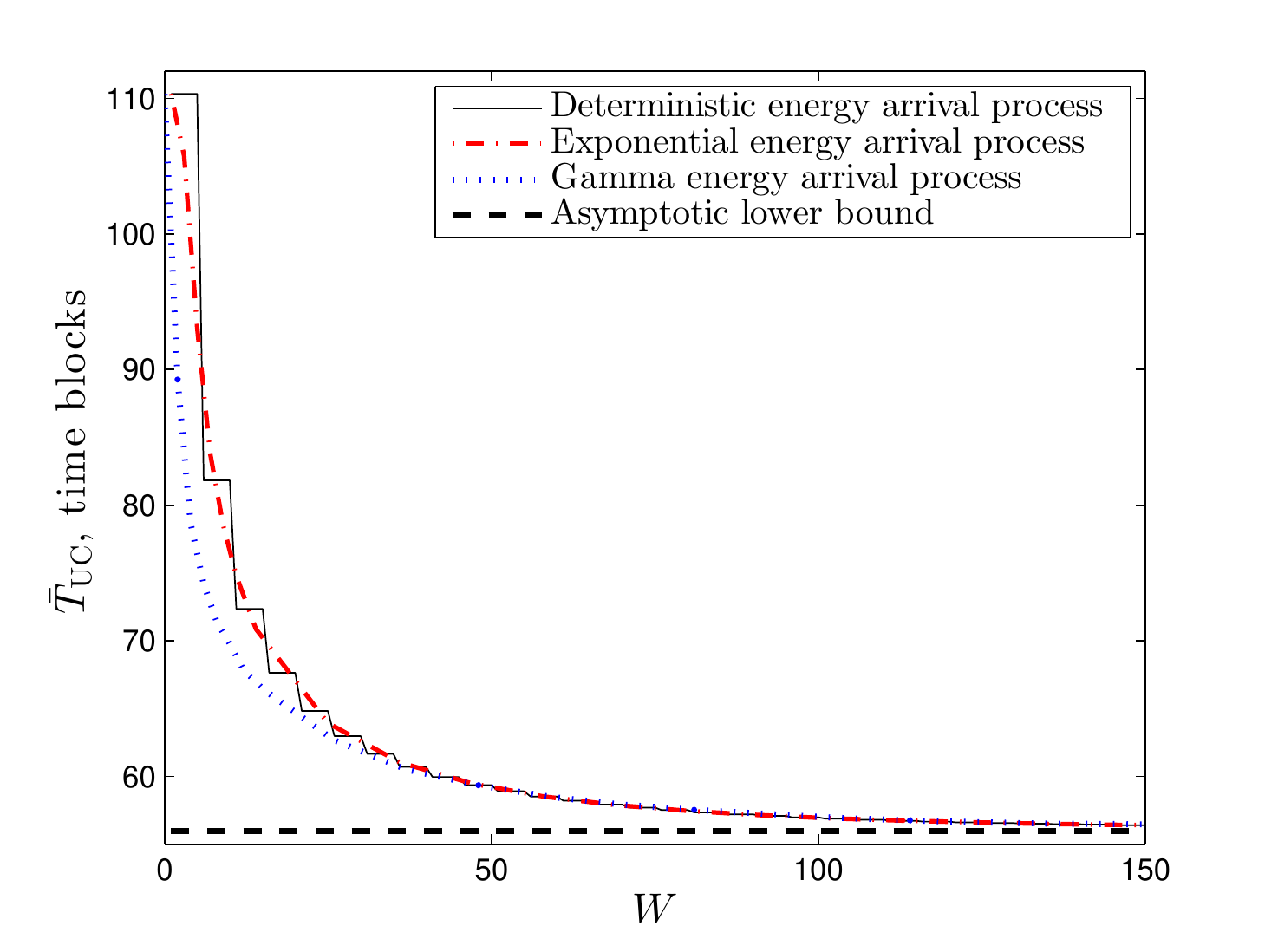}
	\vspace*{-1.3cm}
	\caption{Average update cycle, $\aveTuc$, versus $W$, with different energy arrival processes.}\label{fig:diff_distribution_UC}
	\endminipage\\
	\hrulefill
	\vspace*{-0.7cm}
\end{figure*}

	{\textbf{Average update age and average update cycle with different average harvested power:}
		Fig.~8 shows  the average update age, $\aveTstd$, and the average update cycle, $\aveTuc$, for different average harvested power values, $\rho$, with an exponential energy arrival process. The results are plotted using Corollaries 3 and 7. 
		For the update age, we see that when the average harvested power is very low, i.e., less than $-2$~dBm, the update age is one time block. 
		This is expected since sufficiently low average harvested power cannot enable any retransmission during time window $W-1$, i.e., a packet is either successfully transmitted in the first transmission block right after the sensing block (an update age of one) or dropped due to no chance of retransmission.		
		With an increase of average harvested power, retransmissions are enabled, which makes the update age increases beyond one.
		However, when the average harvested power is higher than $8$~dBm, the average update age monotonically decreases with an increase of the average harvested power. This is as expected: the sensor requires fewer energy harvesting blocks to perform retransmissions,
		and hence, the sink is likely to receive the packet in a more timely manner (i.e., with a smaller update age).
		For the update cycle, we see that the average update cycle monotonically decreases with average harvested power. Again, this is expected since a higher average harvested power enables more transmission blocks within a certain time duration, and hence, more successful block transmissions are likely to occur within a given time duration, i.e., the update cycle decreases.
		Also we see that when the average harvested power is very high, i.e., $\rho \geq 30$~dBm,
		both update age and update cycle converge to constant values which can be obtained by letting $\rho\rightarrow \infty$ in Corollaries~3 and 7, respectively.
		\textit{Thus, without changing the parameters of the communication protocol, the improvement in delay performance is limited when increasing the average harvested power.}}
%
	\begin{figure*}[t]
		\renewcommand{\captionfont}{\small} \renewcommand{\captionlabelfont}{\small}
		\minipage{0.49\textwidth}
		\vspace*{-0.6cm}
		\includegraphics[width=\linewidth]{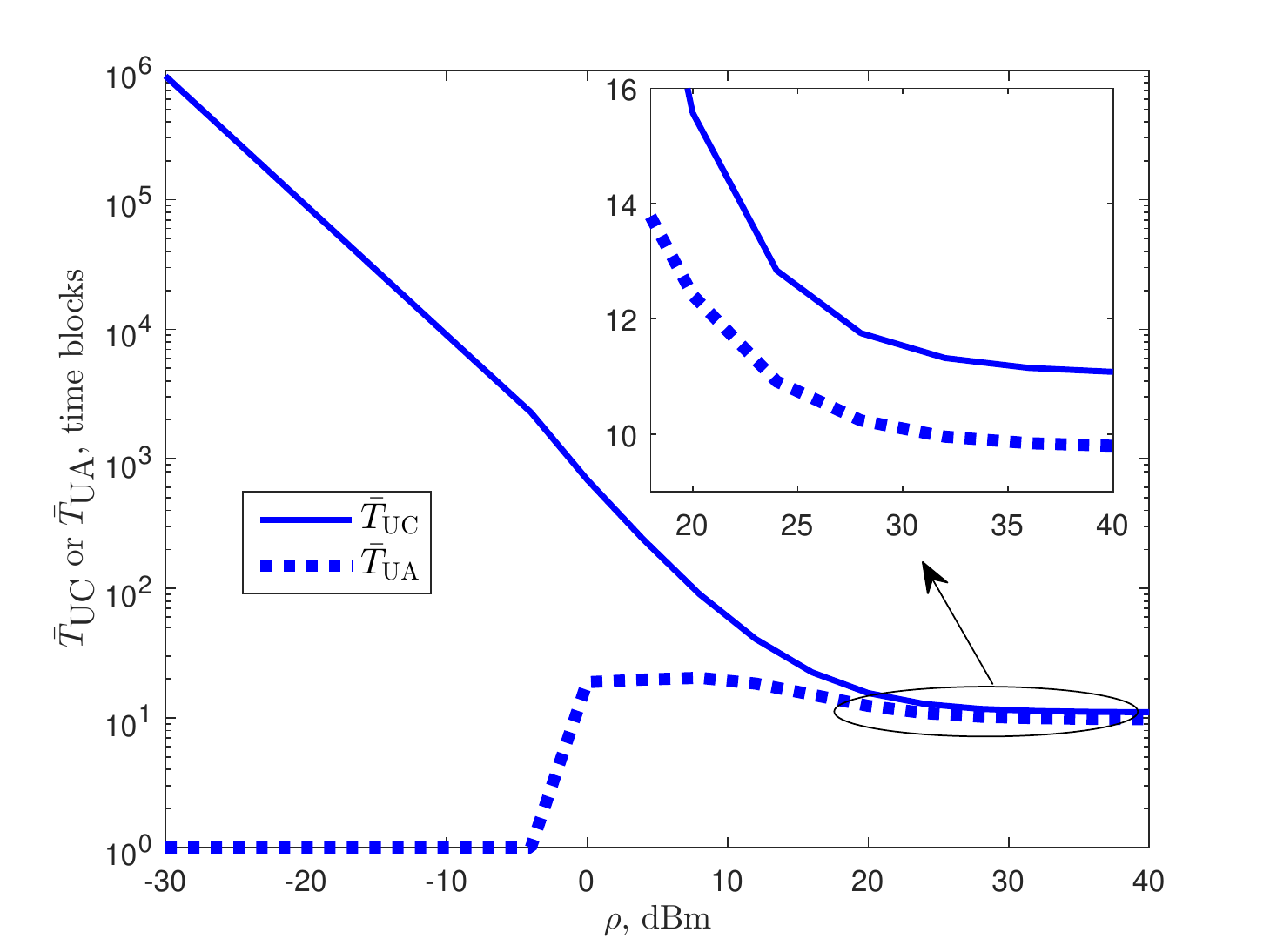}
		\vspace*{-1.3cm}
		\caption{{$\aveTuc$ and $\aveTstd$  versus $\rho$ with exponential energy arrival process.}}\label{fig:diff_EH_rate_STD}
		\endminipage\hfill
		\minipage{0.49\textwidth}
		\vspace*{-0.6cm}
		\includegraphics[width=\linewidth]{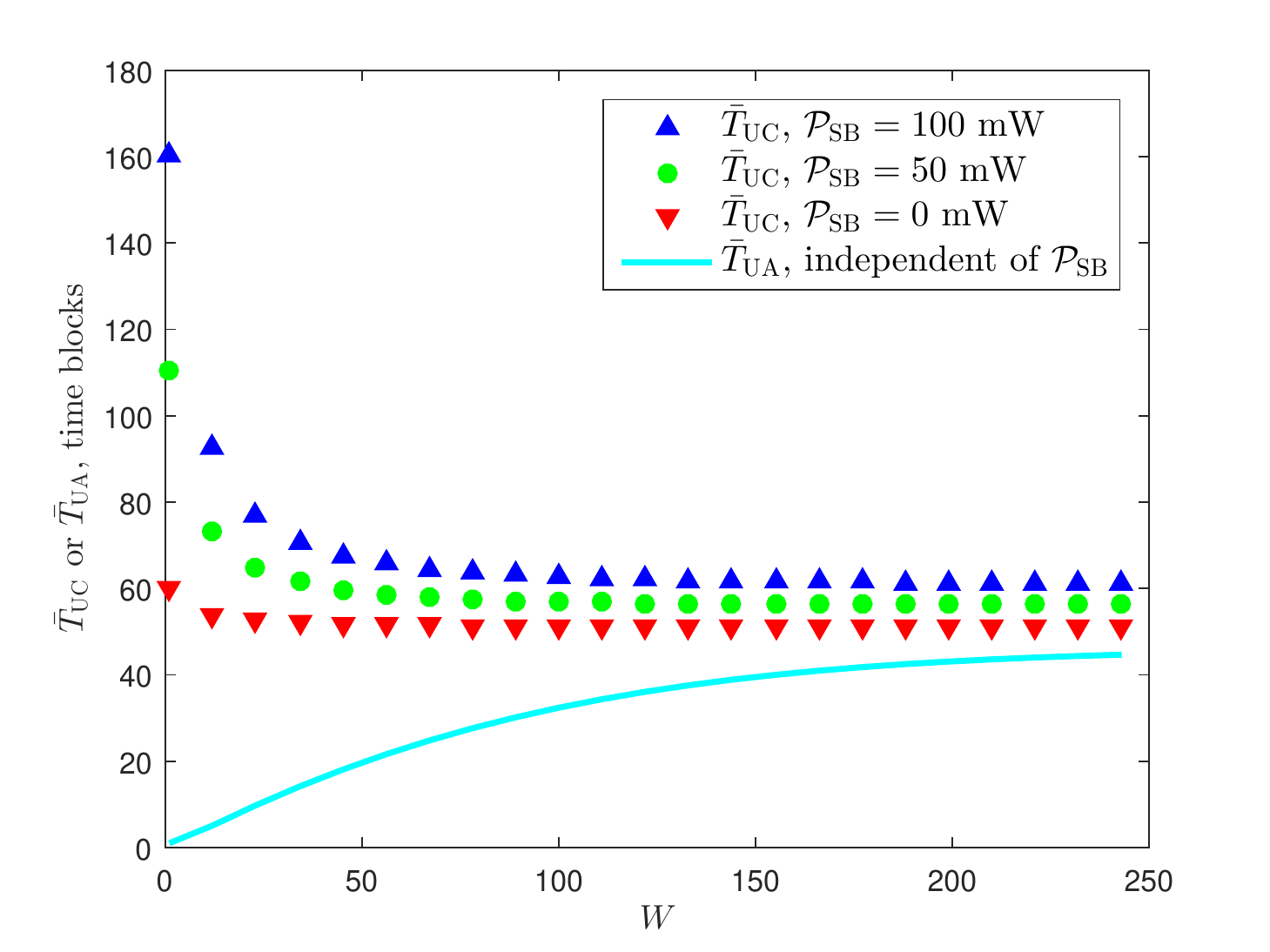}
		\vspace*{-1.3cm}
		\caption{$\aveTuc$ and $\aveTstd$ versus $W$ for different sensing power, $\Psen$.}\label{fig:curve_N}
		\endminipage\\
		\hrulefill
		\vspace*{-0.7cm}
	\end{figure*}	

\textbf{Effect of energy cost of sensing on average update cycle:} We illustrate the effect of energy cost of sensing on average update cycle with exponential energy arrival process as a special case of the random energy arrival process. Fig. \ref{fig:curve_N} shows the average update age, $\aveTstd$,  and the average update cycle, $\aveTuc$, as a function of $W$, with different energy cost of sensing, $\Psen$. The figure shows that the average update age increases as $W$ increases (consistent with Fig.~\ref{fig:diff_distribution_STD}) but it does not change with the energy cost of sensing, {i.e., \textit{the energy cost of sensing has no effect on the update age}}. This is in perfect agreement with our earlier observations and explanations  provided in Remark~\ref{R1}. We can see that for a fixed value of $W$, the average update cycle increases as the sensing power consumption increases from $50$~mW to $100$~mW, {i.e., \textit{the higher the energy cost of sensing the lower update frequency.}} This is in perfect agreement with our earlier observations and explanations provided in Remark~\ref{R2}. To place these results in context with existing studies in the literature that commonly ignore the energy cost of sensing, we also include the result with zero energy cost of sensing. When $\Psen = 0$~mW, we can see that $\aveTuc$ is almost constant around the value of $50$ and does not vary much with $W$.
%

\textbf{Tradeoff between average update age and average update cycle:} Fig.~\ref{fig:region_N} shows the tradeoff between average update age, $\aveTstd$, and average update cycle, $\aveTuc$ with exponential energy arrival process. The different points on the same curve are achieved with different $W$. We can see that when the energy cost of sensing is comparable to or larger than the energy cost of transmission, e.g., $\Psen = 50$~mW and $\Psen = 100$~mW, the reduction in $\aveTstd$ can result in a significant increase in $\aveTuc$, and vice versa. For example, when $\Psen = 100$ mW, decreasing $\aveTstd$ from $15$ to $5$ time blocks, causes the $\aveTuc$ to increase from $75$ to $95$ time blocks. However, when the energy cost for sensing is negligible, e.g., $\Psen = 0$~mW, such a tradeoff is almost barely noticeable. For example, decreasing $\aveTstd$ from $15$ to $5$ time blocks, results in $\aveTuc$ increasing by two time blocks, i.e., a significant change in $\aveTstd$ does not result in a noticeable change in $\aveTuc$. These trends in Fig.~\ref{fig:region_N} are in accordance with our earlier observations in Remark~\ref{R2}.
{\textit{Thus, with the consideration of sensing energy cost, an increase of update frequency is achieved at the expense of update freshness, and vice versa.}}
\begin{figure*}[t]
	\renewcommand{\captionfont}{\small} \renewcommand{\captionlabelfont}{\small}
	\minipage{0.49\textwidth}
	\vspace*{-0.5cm}
	\includegraphics[width=\linewidth]{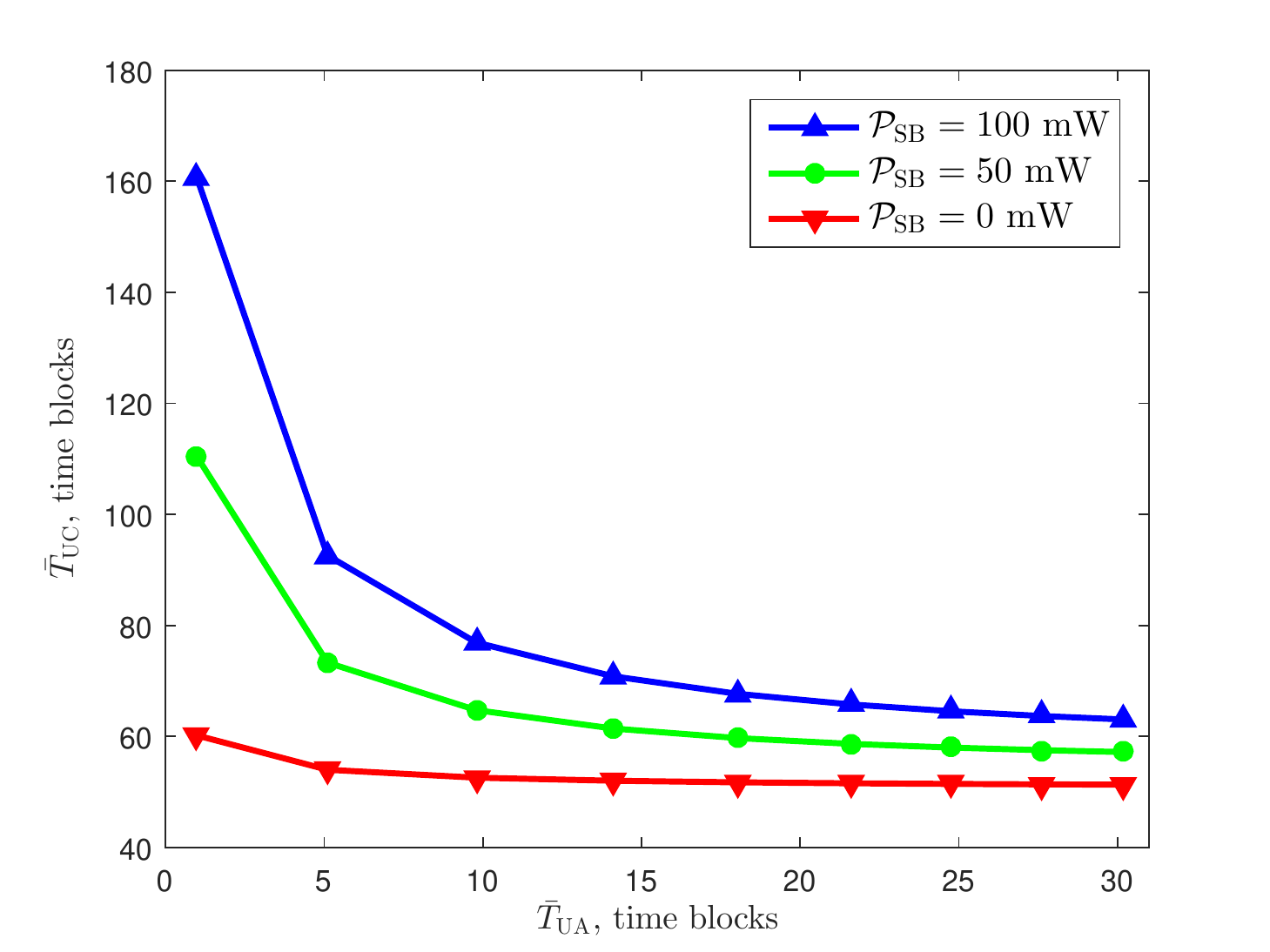}
	\vspace*{-1.2cm}
	\caption{Tradeoff between $\aveTuc$ and $\aveTstd$.}\label{fig:region_N}
	\endminipage\hfill
	\minipage{0.49\textwidth}
	\vspace*{-0.5cm}
	\includegraphics[width=\linewidth]{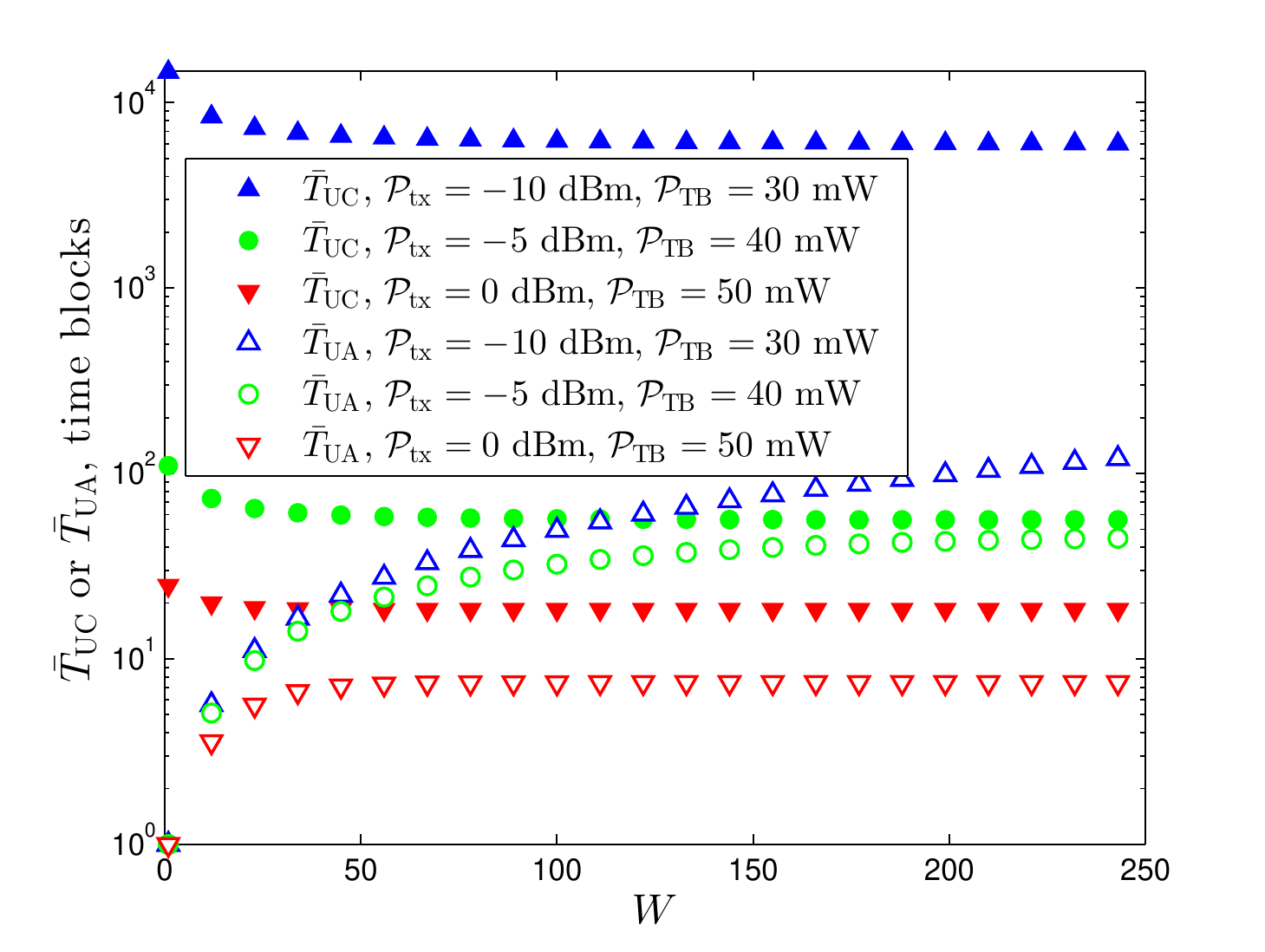}
	\vspace*{-1.2cm}
	\caption{{\!\!$\aveTuc$\! and\! $\aveTstd$\! versus \!$W$ with different $\Prf$\! and\!~$\Ptx$.}}\label{fig:diff_power}
	\endminipage\\
	\hrulefill
	\vspace*{-0.7cm}
\end{figure*}

{\textbf{{Effect of transmit power consumption on average update age and averge update cycle.}}} {Fig. 11} {shows
	the impact of power consumption on $\aveTstd$ and $\aveTuc$, for different values of transmit power $\Ptx$ and RF transmit power $\Prf$, with an exponential energy arrival process. 			 
	In reality, $\Ptx$ and $\Prf$ do not have a linear relationship.
	Three pairs of typical values found in [25] are chosen.
	We see that both $\aveTstd$ and $\aveTuc$ decrease with $\Ptx$ or $\Prf$.
	This is as expected: if the transmit power is small, $\Pout$ is high, resulting in a large number of retransmissions until the sensed information is successfully transmitted or $W-1$ time blocks are reached. 
	As a result, $\aveTstd$ and $\aveTuc$ are large when the transmit power is small.
	\textit{Thus, under these above parameter choices, a higher transmit power results in better delay performance.}}

\section{Conclusions}
This paper has analysed the delay performance of an EH sensor network, focusing on the operation of a single EH sensor and its information transmission to a sink. The energy costs of both sensing and transmission were taken into account. Two metrics were proposed, namely the update age and update cycle. In order to limit the delay due to retransmissions, a time window for retransmissions was imposed. Using both a deterministic and a general random energy arrival model, the exact probability mass functions and the mean values of both metrics were derived. The results showed that the average update age increases while the average update cycle decreases with increasing retransmission window length. The average update age is independent of the energy cost of sensing but the average update cycle increases as the energy cost of sensing increases. In addition, a tradeoff between update age and update cycle was illustrated when the energy cost of sensing is comparable to the energy cost of transmission. 
{Future work can consider the impact of non-deterministic time for receiving the feedback signal at the sensor.}
%
	\renewcommand{\theequation}{\thesection.\arabic{equation}}
	\numberwithin{equation}{section}
%
%
%
%
%
%
%

\section*{Appendix A: Proof of Lemma \ref{lemma_prime}}
\setcounter{equation}{0}
 \renewcommand{\theequation}{A.\arabic{equation}}
We first define the block-wise harvest-then-use process, and then propose and prove Lemma~\ref{lemma_prime}.\\
	{\textbf{Definition A1} (Block-wise harvest-then-use process)\textbf{.} 
	{A harvest-then-use process consists of energy harvesting blocks (EHBs) and energy consumption blocks (ECBs). 
	It starts and keeps on harvesting energy with EHBs. Once the available, i.e., accumulated, energy is no less than a threshold of $Q$ Joules, an ECB occurs, and consumes $Q$ Joules of energy. 
	If this condition for ECB is not satisfied, the process goes back to harvest energy with EHBs.}}
	
{During the harvest-then-use process, the harvested energy in the $i$th EHB is represented by $\xi_i$, $i=1,2,3,...$, and the available energy after the $j$th ECB is represented by $\tilde{\Xi}_j$, $j=1,2,3,...$. 
	Due to the randomness of the energy arrival process, i.e., $\xi_i$ is a random variable, the available energy after each ECB, $\tilde{\Xi}_j$, is also a random variable which only depends on $\xi_i$.
	Furthermore, using the statistics of $\xi_i$, and modeling $\tilde{\Xi}_j$, $j=1,2,3,...$, as a random process, an important feature of the random process is revealed in Lemma A1.}

    \setcounter{lemma}{0}
    \renewcommand{\thelemma}{A\arabic{lemma}}
\begin{lemma} \label{lemma_prime}
	\textnormal{
	For block-wise harvest-then-use process with energy threshold $Q$, where the harvested energy in each EHB, $\xi_i$, $i=1,2,3,...$, are independent and identically distributed, each with pdf {containing at least one} positive right-continuous point, $f(x)$, the available energy after each ECB, $\tilde{\Xi}_j$, $j=1,2,3,...,$, consists of a positive recurrent Harris chain, with unique steady-state distribution which is given by
	\begin{equation} \label{first_gx}
	g \left(x \right) = \frac{1}{\rho} \left(1- F(x)\right),
	\end{equation}
	where $F(x)$ and $\rho$ are respectively, the cdf and the mean of $\xi_i$.		
		}
\end{lemma}
\begin{proof}
	The proof consists of two steps.
	In the first step, we prove that the energy state after the $j$th ECB, $\tilde{\Xi}_j$, constitutes a positive recurrent Harris chain (a collection of Markov chains with uncountable state space). Thus, a unique steady-state distribution of $\tilde{\Xi}_j$ exists \cite{Book_probability}.
	In the second step, we prove that \eqref{first_gx} is the unique steady-state distribution.
	\par
	{\bf Step 1:} It is easy to see that the current state, $\tilde{\Xi}_j$ takes its value from a continuous state space and only relies on the previous energy state $\tilde{\Xi}_{j-1}$, thus $\tilde{\Xi}_j$, $j=1,2,3,...$, forms a continuous-state Markov chain.
	Without loss of generality, we assume that $\sup\left\lbrace{\xi_i}\right\rbrace = B$, thus $\sup\left\lbrace \tilde{\Xi}_j \right\rbrace \leq B$ holds in this harvest-then-use process.\footnote{Note that although we assume $B$ is finite, the infinite case can be easily generated from the discussions below, thus is omitted due to space limitation.}
	It is easy to see that the state space of Markov chain $\tilde{\Xi}_j$, $\mathcal{S}$, is a subset of $[0,B)$, and because of the harvest-then-use protocol, any current state which is higher than $Q$, will access the interval $[0,Q)$ in the following steps.
	Thus, we only need to prove that any state $s \in [0,\min\lbrace B,Q\rbrace)$ can hit any arbitrary small interval $\bm \tau = (\tau^-,\tau^+)$ in $\mathcal{S}$ with non-zero probability within finite steps. Actually, in the following, we complete the proof with the assumption that $\mathcal{S}=[0,B)$, which also proves that the state space of Markov chain $\tilde{\Xi}_i$ is exactly $[0,B)$.
	
	In the following, using a constructive method, we show that for Markov chain $\tilde{\Xi}_j$, given any current state $s\in [0,\min\lbrace B,Q\rbrace)$, there is at least a probability, $q\times p$, that any arbitrary small interval $\bm \tau$ will be accessed with $\tilde{j}$ steps, where $p$, $q$, $\tilde{j}$ are defined below which only depends on the state $s$, the interval length $\tau$ and the pdf of the harvested energy in each EHB.

{Since pdf function $f(x)$ has positive right-continuous points on $[0,B)$, there exists at least one interval $[D^-,D)$ that satisfies
	$[D^-,D) \subset [0,B)$, $D-D^- = \tau/2$, and $f(x)$ is positive right-continuous on $[D^-,D)$.
	We assume that $D^- \geq \tau^+$, and the $D^- < \tau^+$ case can be easily generated from the the following discussions, thus is omitted due to space limitation.
	Now we define $p\triangleq\int_{D^-}^{D}f(x) \mathrm{d}x$ as the probability that harvested energy in one EHB lies in the interval $[D^-,D)$.
	Also we define $\tilde{f}(x) =f(x)$ when $x\in [D^-,D)$, otherwise $\tilde{f}(x) =0$,
	and $\tilde{f}_i(x)$ is the $i$-fold convolution of function $\tilde{f}(x)$.
	
	Thus, it is easy to see that $\tilde{f}_i(x)$ is positive and continuous in the interval $(iD^-,iD)$,  and $\int_{a}^{b}\tilde{f}_i(x) \mathrm{d}x$ is the probability that the harvested energy by $i$ EHBs lies in the interval $[a,b)$, while the energy harvested by each of the $i$ EHBs lies in the interval $[D^-,D)$.}
	Thus, letting 
	{$\tilde{i}\triangleq {\lceil 4(Q+\tau)/\tau -1 \rceil}$ and 
	$\tilde{j} \triangleq\lfloor ((\tilde{i}+1)D-\tau^+)/Q \rfloor$,}
	given the current energy state $s$, after $\tilde{i}$ EHBs,
	the \emph{accumulated energy level} lies in the interval $\mathcal{A}\triangleq(s+\tilde{i}D^-,s+\tilde{i}D)$ with positive probability distribution.
	Also we see that interval
	$\Delta \triangleq (\tilde{j}Q+\tau^- -D^-,\tilde{j}Q+\tau^+ -D) \subset \mathcal{A}$,
	thus, there is at least (because we have only considered the scenario that harvested energy by each EHB lies in $[D^-,D)$)
	a probability
	$q\triangleq \inf\lbrace\int_{\tilde{\bm \tau}}\tilde{f}_{\tilde{i}}(x), \text{ interval }\tilde{\bm \tau} \subset \mathcal{S}, \text{length of }\tilde{\bm \tau}=\text{length of }\Delta =\tau/2\rbrace$ that the accumulated energy level lies in $\Delta$.
	Therefore,
	after the next EHB with probability $p$ that the harvested energy lies in $[D^-,D)$, the accumulated energy level lies in the interval $[\tilde{j}Q+\tau^-,\tilde{j}Q+\tau^+)$,
	which means that after the current state $\tilde{\Xi}_j = s$, with $\tilde{j}$ steps (each step consumes the amount of energy, $Q$), there is at least a probability, $q\times p$ to
	make the Markov chain hit the interval $(\tau^-,\tau^+)$.
	Thus, Markov chain $\tilde{\Xi}_j$ is a positive recurrent Harris chain \cite{Book_probability}.

	\par
	{\bf Step 2:} In the aforementioned Markov chain, we still assume that the current state $\tilde{\Xi}_j = s$. Thus, in the previous state, the available energy could be higher than $Q$, i.e., $\tilde{\Xi}_{j-1} = s+Q$,
	and $\tilde{\Xi}_{j-1}$ could also be smaller than $Q$, i.e., based on energy level $\tilde{\Xi}_{j-1}$, there are $i$ EHBs ($i=1,2,3,...$) to make the energy level reach $Q+s$, which makes $\tilde{\Xi}_{j} = s$.
	Based on the above and the Markovian property, the steady-state distribution of the process, $g(x)$, should satisfy the following conditions:
	(1) $
			\int_{0}^{\infty} g(x) =1
			$ and 
	(2) $
			g(x) = g(x+Q) + \sum\limits_{i=1}^{\infty} \int_{x}^{Q+x} g_{i-1}(Q+x-y) f(y) \mathrm{d}y
			$.
	where $g_i(x)$ represent the pdf of energy level after $i$ EHBs following a ECB, which is given by
	\begin{equation} \label{eq:g_i}
	g_i(x) = \left\lbrace
	\begin{aligned}
	&g(x), &i =0,\\
	&\left(g\underbrace{\star f \star f\star ... \star f}_{i \text{ convolutions}}\right)\!\!(x), &i >0.
	\end{aligned}
	\right.
	\end{equation}
	Because $f(x)$ and $g_i(x) \geq 0$ for all $x$ and $i=0,1,2,...$,  by using Tonelli's theorem for sums and integrals \cite{Terry}, we exchange the summation and integral operator in Condition 2, thus we have
	\begin{equation} \label{condition_2_new}
	g(x) = g(x+Q) + \int\limits_{x}^{Q+x} \left(\sum\limits_{i=0}^{\infty} g_i(Q+x-y)\right) f(y) \mathrm{d}y.
	\end{equation}

	{Taking \eqref{first_gx} into \eqref{eq:g_i}, we have 		
	\begin{equation} \label{new_g(x)}
	g_i(x) = \frac{1}{\rho}\left(\left(F \underbrace{\star f \star f ... \star f}_{i-1\text{ convolutions}}\right)(x) - \left(F\underbrace{\star f\star f ... \star f}_{i\text{ convolutions}}\right)(x)\right),\ i >0.
	\end{equation}
	Since $0 \leq F(x) \leq 1 $, $f(x) \geq 0$ and $\int_{0}^{\infty}f(x)=1$, when $i\rightarrow \infty$, we have~\cite{Book_probability}
	\begin{equation} \label{limitation}
	\left(F\underbrace{\star f\star f ... \star f}_{i\text{ convolutions}}\right)(x) \rightarrow 0.
	\end{equation}
	From \eqref{new_g(x)} and \eqref{limitation}, we have		
	\begin{equation} \label{my_summation}
	\begin{aligned}
	&\sum\limits_{i=0}^{\infty} g_i(Q+x-y) 
	= \frac{1}{\rho}\left(1-F(Q+x-y)\right) 
	+ \frac{1}{\rho}\left(F(Q+x-y) - \left(F \star f\right)(Q+x-y)\right)	\\
	&\!+\!\frac{1}{\rho}\left(\left(F \star f\right)(Q\!+\!x\!-\!y)\! -\! \left(F \star f \star f\right)(Q\!+\!x\!-\!y)\right) \!+\! ...
	\!= \frac{1}{\rho} \left(1 \!-\! \lim_{i\rightarrow \infty} \!\left(\!F\underbrace{\star f\star f ... \star f}_{i\text{ convolutions}}\!\right)\!(x) \right) \!=\! \frac{1}{\rho}.
	\end{aligned}	
	\end{equation}
	Taking \eqref{my_summation} and \eqref{first_gx} into the right side of \eqref{condition_2_new}, we have 
	\begin{equation}
	\begin{aligned}
	&g(x+Q) + \int\limits_{x}^{Q+x} \left(\sum\limits_{i=0}^{\infty} g_i(Q+x-y)\right) f(y) \mathrm{d}y=
	\frac{1}{\rho} (1-F(x+Q)) + \int\limits_{x}^{Q+x} \frac{1}{\rho}f(y) \mathrm{d}y\\
	&=\frac{1}{\rho} (1-F(x+Q))  + \frac{1}{\rho} \left(F(Q+x) - (F(x)\right)
	= \frac{1}{\rho} \left(1 - F(x)\right).
	\end{aligned}
	\end{equation}
	Thus, $g(x)$ in \eqref{eq:g_i} satisfies Condition 2. }
	Because of $\int_{0}^{\infty} (1-F(x)) = \myexpect{\xi_i}$ \cite{Book_probability}, Condition 1 is also satisfied, yielding the desired result.
\end{proof}

\section*{Appendix B: Proof of Lemma 1}
\setcounter{equation}{0}
 \renewcommand{\theequation}{B.\arabic{equation}}
For general random energy arrival processes, the proof is based on Lemma \ref{lemma_prime} given in Appendix A.
First, we find an arbitrarily small $Q$ which is a constant such that $\Esen$ and $\Etx$ are integer multiples of it.
Then, from an energy perspective, we equivalently treat the proposed communication protocol with energy harvesting, sensing and transmission as a simple harvest-then-use process with EHBs and ECBs (each consumes energy, $Q$) as discussed in Lemma \ref{lemma_prime}. Thus, the energy level after a TB, can be treated equivalently as that after a corresponding ECB.
Therefore, the steady-state distribution of energy level after each TB is the same as that after each ECB, which is given in Lemma \ref{lemma_prime}, completing the proof.

%

\section*{Appendix C: Event and Random Variable Definitions}
\setcounter{equation}{0}
 \renewcommand{\theequation}{C.\arabic{equation}}
To assist the proofs of the main results,
we use $\UC$ to denote the sequence of time blocks from an arbitrary STB to the next STB.
Also we define two events (according to  \cite{Book_probability}) and several discrete random variables (r.v.s) for convenience:
\begin{enumerate}[1)]
	\item Event $\eventsuc$: Given a SB, its generated information is successfully transmitted to the sink, i.e., STB occurs during the $W$ blocks after the SB.	
	\item Event $\eventfail$: Given a SB, its generated information is not successfully transmitted to the sink, i.e., STB does not occur during the $W$ blocks after the SB.

	\item r.v. $N$, $1 \leq N \leq W$: Given a SB, it is followed by $N$ TBs before the next SB.
	I.e., if $\eventsuc$ occurs, the $N$ TBs includes $N-1$ FTBs and one STB. While if $\eventfail$ occurs, all the $N$ TBs are FTBs.
	\item r.v. $L$, $1 \leq L \leq W$: After a SB, the $L$th block is the last TB before the next SB.
	I.e., if $\eventsuc$ occurs, the $L$th block is a STB, thus $L$ is the update age. While if $\eventfail$ occurs, the $L$th block is the last FTB during the time window for retransmissions, $W$.
	\item r.v. $\tilde{V}$, $\tilde{V} \geq -1$: Given a SB,
	if $\eventsuc$ occurs, $\tilde{V} =-1$, while if
	an $\eventfail$ occurs, $\tilde{V}$ is the number of the required EHBs after the time window for retransmissions, $W$, in order to harvest the amount of energy, $\Etx$.
	Note that, after a $\eventfail$, the amount of energy $\Esen + \Etx$ is required to be reached in order to support the following SB and TB.
	Without loss of generality, here we assume that the energy harvesting process first meets the energy level $\Etx$, and the TB consumes the energy, $\Etx$, (V)irtually.
	From Lemma 1 and its proof, the steady-state distribution of the available energy level after the $\tilde{V}$ EHBs is $g(\epsilon)$.
	\item r.v. $V$, $V \geq 0$. Given a SB and conditioned on a $\eventfail$ occurs, $V$ is the number of the required EHBs after the time window for retransmissions, $W$, in order to harvest the amount of energy, $\Etx$. From the definition of $V$ and $\tilde{V}$, it is easy to see that
	\begin{equation} \label{V_and_V}
	\myprobability{V=v} = \myprobability{\tilde{V}=v \vert \eventfail} = \frac{\myprobability{\tilde{V}=v}}{\myprobability{\eventfail}}
	, v=0,1,2,....
	\end{equation}

	\item r.v. $E(\mathcal{E})$, $E(\mathcal{E}) \geq 0$: Given the distribution of initial energy level, $g(\epsilon)$,
	and the amount of target energy, $\mathcal{E}$, the required number of energy harvesting block is $E(\mathcal{E})$.
	
		For a deterministic energy arrival process, straightforwardly we have
		\begin{equation} \label{det_prob_EH}
		\myprobability{E(\mathcal{E}) =i} = 1, i= \mathcal{E}/\rho.
		\end{equation}
		
		For a general random energy arrival process, from the definition of $E(\mathcal{E})$, Lemma \ref{L1} and its proof, we have
		\begin{equation} \label{def_prob_EH}
		\myprobability{E(\mathcal{E}) = i} = G_{i-1}(\mathcal{E})-G_{i}(\mathcal{E}), \ i= 0,1,2,...
		\end{equation}
		where
		\begin{equation} \label{G_x}
		G_{i}(x) =
		\left\lbrace
		\begin{aligned}
		&1, &i=-1,\\
		&\int_{0}^{x} g_i(u) \mathrm{d}u, &i\geq 0,
		\end{aligned}
		\right.
		\end{equation}	
		and $g_i(x)$ is defined in \eqref{eq:g_i}.

		\noindent For \emph{exponential} energy arrival process,
		we know that the energy accumulation process during EHBs after a TB is a \emph{Poisson} process \cite{Book_probability}, thus, we have
		\begin{equation} \label{prob_EHB_exp}
		\myprobability{E(\mathcal{E}) = i} =G_{i-1}(\mathcal{E})-G_{i}(\mathcal{E})
		= \Pois{i}{\mathcal{E}/\rho}
		= \frac{(\mathcal{E}/\rho)^i e^{-\mathcal{E}/\rho}}{i!}, \ i= 0,1,2,...
		\end{equation}

	\item r.v. $M$, $M \geq 0$: Given a UC, $\eventfail$ occur $M$ times and followed by one $\eventsuc$ in it.
	\end{enumerate}
	
	From the definitions of event, we know that $\eventsuc$ and $\eventfail$ are mutually exclusive events.
	Thus, we have
	\begin{equation} \label{def_Psuc}
	\Psuc \triangleq \myprobability{\eventsuc} \text{ and } \myprobability{\eventfail} = 1-\Psuc,
	\end{equation}
	where $\eventsuc$ and $\eventfail$ depends on transmit outage probability in each TB, and the available energy after the first TB following the SB.
	Because we have assumed that the success of each transmission are independent of one another, and from Lemma 1, the distribution of the available energy after each TB is the same,
	each event $\eventsuc/\eventfail$ is independent with each other during the communication process.
	Therefore, r.v. $M$ follows the geometric distribution
	\begin{equation} \label{M_geo}
	\myprobability{M=m} = \Psuc \left(1-\Psuc\right)^m, \ m=0,1,2,....
	\end{equation}

	\section*{Appendix D: Pmf of Update Age}	
	\setcounter{equation}{0}
	 \renewcommand{\theequation}{D.\arabic{equation}}
	From the definitions in Appendix C,
	the pdf of $\Tstd$ can be calculated as
	\begin{equation} \label{pmf_STD_basic_1}
	\myprobability{\Tstd = k} = \frac{\myprobability{L = k, \eventsuc}}
	{\myprobability{\eventsuc}},\ k=1,2,...,W.
	\end{equation}
	Using the law of total probability and the r.v.s defined in Appendix C, we have
	\begin{equation} \label{pmf_STD_basic_2}
	\begin{aligned}
	\myprobability{L = k, \eventsuc}
	&=\! \sum\limits_{n=1}^{k} \myprobability{L\!=\!k,N\!=\!n,\eventsuc}
	\!=\! \sum\limits_{n=1}^{k} \myprobability{N\!=\!n,E((n-1)\Etx) \!=\! k-n, \eventsuc}
	\\
	&= \sum\limits_{n=1}^{k} \myprobability{N\!=\!n,\eventsuc \vert E((n-1)\Etx)\! =\! k-n} \myprobability{E((n-1)\Etx) \!=\! k-n} \\
	&= \sum\limits_{n=1}^{k} (1-\Pout)\left(\Pout\right)^{n-1} \myprobability{E((n-1)\Etx) = k-n}.
	\end{aligned}
	\end{equation}

	\noindent Again using the law of total probability and using \eqref{pmf_STD_basic_2}, \eqref{def_Psuc} becomes
	\begin{equation} \label{pmf_STD_basic_2_2}
	\begin{aligned}
	\Psuc
	&= \myprobability{\eventsuc} = \sum\limits_{l=1}^{W} \myprobability{L=l, \eventsuc}
	= \myprobability{L=1, \eventsuc} + \sum\limits_{l=2}^{W} \myprobability{L=l, \eventsuc}\\
		&=1-\Pout + \sum\limits_{l=2}^{W} \sum\limits_{n=2}^{l} \myprobability{L=l, N=n, \eventsuc}\\	
		&=1-\Pout + \sum\limits_{l=2}^{W} \sum\limits_{n=2}^{l} \myprobability{E((n-1)\Etx)=l-n, N=n, \eventsuc}\\
		&=1-\Pout + \sum\limits_{l=2}^{W} \sum\limits_{n=2}^{l} \myprobability{N=n, \eventsuc \vert E((n-1)\Etx)=l-n} \myprobability{E((n-1)\Etx)=l-n}\\
	\end{aligned}
	\end{equation}
	\begin{equation*}
	\begin{aligned}	
	&=1-\Pout + \sum\limits_{l=2}^{W} \sum\limits_{n=2}^{l}
	(1-\Pout) \left(\Pout\right)^{n-1} \ \myprobability{E((n-1)\Etx)=l-n}.\\
	\end{aligned}
	\end{equation*}

	By taking \eqref{det_prob_EH}, \eqref{def_prob_EH} and \eqref{prob_EHB_exp} into \eqref{pmf_STD_basic_2} and \eqref{pmf_STD_basic_2_2}, and then substituting \eqref{pmf_STD_basic_2} and \eqref{pmf_STD_basic_2_2} into \eqref{pmf_STD_basic_1},
	the pmfs of $\Tstd$ for deterministic, general random and exponential energy arrival process are given in Theorems 1, 2 and 3, respectively.

\section*{Appendix E: Pmf of Update Cycle}
\setcounter{equation}{0}
 \renewcommand{\theequation}{E.\arabic{equation}}
First, assuming that $\eventfail$ occurs $m$ times during a UC, we define r.v.s $E_0$, $V_i$, $E_i$, $i=1,2,...,m$, and $\tilde{L}$.
$E_0$ is the number of EHBs at the beginning of the UC until the first SB occurs which follows the same pmf with r.v. $E(\Esen +\Etx)$.
$V_i$ is the number of EHBs required to harvest the amount of energy, $\Etx$, outside the time window for retransmissions of the $i$th $\eventfail$.
$E_i$ is the number of EHBs required to harvest the amount of energy, $\Esen$, following $V_i$ EHBs after the $i$th $\eventfail$.
$\tilde{L}$ is the number of blocks after a SB to the last TB before the next SB, and the TB is a STB.
From the r.v. definitions in Appendix C, $E_0$, $V_i$ and $E_i$ follow the same distribution with r.v.s $E(\Esen +\Etx)$, $V$ and $E({\Esen})$, respectively, and
\begin{equation}
\myprobability{\tilde{L}=l} =\myprobability{L=l,\eventsuc},\ l=1,2,...,W.
\end{equation}
From Lemma 1, $E_0$, $V_i$, $E_i$, $i=1,2,...,m$, and $\tilde{L}$ are mutually independent.

Then, the pmf of update cycle can be calculated as
{
\begin{equation} \label{pmf_UC_1}
\begin{aligned}
&\myprobability{\Tuc = k}
= \sum\limits_{m}^{}\myprobability{\Tuc = k, M=m}\\
&= \sum\limits_{m}^{}\myprobability{E_0 \!+\!E_1\!+\!...\!+\!E_m \!+\!\tilde{V}_1\!+\!...\!+\!\tilde{V}_m \!+\!m\times(1\!+\!W) \!+\! \tilde{L} \!+\!1 =k, \tilde{V}_1,\ \tilde{V}_2,\cdots,\ \tilde{V}_m\geq 0}\\
&= \sum\limits_{m=0}^{\hat{m}}\myprobability{E_0 \!+\!E_1\!+\!...\!+\!E_m \!+\!\tilde{V}_1\!+\!...\!+\!\tilde{V}_m \!+\! \tilde{L} =k -m\times(1\!+\!W) -1, \tilde{V}_1,\ \tilde{V}_2,\cdots,\ \tilde{V}_m\!\geq\! 0}\!,
\\
&\hspace{14cm}
k=2,3,...,
\end{aligned}
\end{equation}}
\par
\vspace{-0.4cm}
\noindent where
$
\hat{m}= \left\lfloor \frac{k-2}{W+1} \right\rfloor.
$
For simplicity, we define the following discrete functions:
\begin{equation} \label{def_3_funcs}
\begin{aligned}
&\zeta(\mathcal{E},i) \triangleq \myprobability{E(\mathcal{E}) =i},\ i=0,1,2,...\\
&\iota(l) \triangleq \myprobability{\tilde{L}=l} = \myprobability{L=l, \eventsuc},\ l=1,2,...,W\\
&\vartheta(v) \triangleq \myprobability{\tilde{V}=v}, \ v=0,1,....
\end{aligned}
\end{equation}
where $\zeta(\mathcal{E},i)$ and $\iota(l)$ are obtained directly from \eqref{det_prob_EH}, \eqref{def_prob_EH}, \eqref{prob_EHB_exp} and \eqref{pmf_STD_basic_2}, respectively,
and $\vartheta(v)$ will be derived later.
Therefore, pmf of $\Tuc$ in \eqref{pmf_UC_1} can be calculated~as
\begin{equation}  \label{Tuc_conv_defi}
\myprobability{\Tuc = k} = \sum\limits_{m=0}^{\hat{m}}
\left(
\zeta(\mathcal{\Esen+\Etx})
\underbrace{\ast\zeta(\mathcal{{\Esen}})\ast \cdots \ast\zeta(\mathcal{{\Esen}})}_{m \text{ convolutions}}
\underbrace{\ast\vartheta\ast \cdots \ast\vartheta}_{m \text{ convolutions}} \ast \iota
\right)
\!\!(k-m(1+W)-1).
\end{equation}

Now we derive the expression for $\vartheta(i)$.
From the definitions of r.v. in Appendix C, we have
{
\begin{equation} \label{func_v}
\begin{aligned}
\vartheta(v)
&=\myprobability{\tilde{V}=v} =\myprobability{\eventfail,\tilde{V}=v} = \sum\limits_{n=1}^{W} \myprobability{\eventfail,\tilde{V}=v, N=n}\\
&= \myprobability{\eventfail,\tilde{V}=v, N=1} + \sum\limits_{n=2}^{W}\myprobability{\eventfail,\tilde{V}=v, N=n}\\
&= \myprobability{\eventfail,\tilde{V}=v, N=1} + \sum\limits_{l=2}^{W}\sum\limits_{n=2}^{l}\myprobability{\eventfail,\tilde{V}=v, N=n,L=l}\\
&= \myprobability{\eventfail,N=1, E(\Etx) = W+v-1}\\
&+ \sum\limits_{l=2}^{W}\sum\limits_{n=2}^{l}
\myprobability{\eventfail,N=n, E((n-1)\Etx) = l-n, E(\Etx) = W+v-l}\\
&= \myprobability{\eventfail,N=1 \vert E(\Etx) = W+v-1} \myprobability{ E(\Etx) = W+v-1}\\
&+ \sum\limits_{l=2}^{W}\sum\limits_{n=2}^{l}
\myprobability{\eventfail,N=n \vert E((n-1)\Etx) = l-n, E(\Etx) = W+v-l} \times\\
& \hspace{5cm}\myprobability{E((n-1)\Etx) = l-n, E(\Etx) = W+v-l}\\
&= \Pout \myprobability{ E(\Etx) = W+v-1}\\
&\hspace{1.3cm}+ \sum\limits_{l=2}^{W}\sum\limits_{n=2}^{l} \left(\Pout\right)^n
\myprobability{E((n-1)\Etx) = l-n} \myprobability{E(\Etx) = W+v-l}.
\end{aligned}
\end{equation}}

By taking functions \eqref{func_v}, $\zeta(\mathcal{E},i)$ and $\iota(l)$ in \eqref{def_3_funcs}, into \eqref{Tuc_conv_defi}, and letting \eqref{det_prob_EH} and \eqref{def_prob_EH} substitute $\myprobability{E(\mathcal{E})=i}$, the pmf of $\Tuc$ for deterministic and general random energy arrival process can be calculated, respectively, as given in Theorems 4 and 5.
While for the exponential energy arrival process, by using the sum property of Poisson distribution, we have
\begin{equation}
\myprobability{E(\mathcal{E}_1)_1 +E(\mathcal{E}_2)_2 = i} = \myprobability{E(\mathcal{E}_1+\mathcal{E}_2)= i},
\end{equation}
where $E(\mathcal{E}_1)_1$ and $E(\mathcal{E}_2)_2$ are two independent random variables which have the same distribution with $E(\mathcal{E}_1)$ and $E(\mathcal{E}_2)$ defined in Appendix C, respectively.
Therefore, letting \eqref{prob_EHB_exp} substitute $\myprobability{E(\mathcal{E})=i}$, the pmf of $\Tuc$ for exponential energy arrival process can be further simplified as given in Theorem 6.

\section*{Appendix F: Average Update Cycle}
\setcounter{equation}{0}
 \renewcommand{\theequation}{F.\arabic{equation}}
Based on Appendix E, average update cycle can be calculated as
\begin{equation} \label{ave_Tuc}
\begin{aligned}
\aveTuc
&= \myexpect{\myexpect{\Tuc \vert M}}
=\sum\limits_{m=0}^{\infty} \myprobability{M=m} \myexpect{\Tuc \vert M=m} \\
&=\sum\limits_{m=0}^{\infty} \myprobability{M=m}
\myexpect{E_0\!+\!E_1\!+\! \cdots\!+\!E_m \!+\!V_1\!+\!V_2\!+\! \cdots\!+\!V_m \!+\! m\times (1\!+\!W) \!+\! 1\!+\!\Tstd} \\
&=\sum\limits_{m=0}^{\infty}
\myprobability{M=m}
\left(
\myexpect{E_0} \!+\! \myexpect{E_1} \!+\!\cdots\!+\!\myexpect{E_m} \!+\! m\times \bar{V} \!+\! m\times (W\!+\!1) \!+\!\aveTstd\!+\!1
\right).
\end{aligned}
\end{equation}
From Appendix E, we have
\begin{equation} \label{ave_E}
\myexpect{E_0} = \frac{\Esen +\Etx}{\rho},\ \myexpect{E_i} =\frac{{\Esen}}{\rho}, i=1,2,...,m.
\end{equation}
After taking \eqref{func_v} and \eqref{def_Psuc} into \eqref{V_and_V} and some simplifications, the expectation of $V$ can be calculated~as
\begin{equation} \label{ave_V}
\begin{aligned}
\bar{V} &= \sum\limits_{v=0}^{\infty} v \frac{\vartheta(v)}{1-\Psuc}=\frac{\Pout}{1-\Psuc} \left(\frac{\Etx}{\rho} \!-\! \sum\limits_{i=0}^{W-2}i\myprobability{E(\Etx) =i} \!-\! (W\!-\!1) \!\left(\!\!1\!-\!\!\!\sum\limits_{i=0}^{W-2}\!\myprobability{E(\Etx) \!=\!i}\!\right) \!\right)\\
		&+\frac{1}{1-\Psuc}
		\sum\limits_{l=2}^{W} \sum\limits_{n=2}^{l}
		\left(\Pout\right)^n \myprobability{E((n-1)\Etx)=l-n}\times\\
		&\hspace{2.5cm}\left(\frac{\Etx}{\rho} - \sum\limits_{i=0}^{W-l-1}i\myprobability{E(\Etx) =i} - (W-l) \left(1-\sum\limits_{i=0}^{W-l-1}\myprobability{E(\Etx) =i} \right) \right).
\end{aligned}
\end{equation}

By taking \eqref{ave_E}, \eqref{ave_V} and \eqref{M_geo} into \eqref{ave_Tuc}, and further substituting $\Psuc$ and $\aveTstd$ given in Corollaries 1, 2 and 3, average update cycle for deterministic, general random and exponential energy arrival processes are given in Corollaries 5, 6 and 7, respectively.

\section*{Appendix G: Asymptotic Lower/Upper Bounds}
\setcounter{equation}{0}
 \renewcommand{\theequation}{G.\arabic{equation}}
From Corollaries 1 and 2, it is easy to see that $\aveTstd$ increase with $W$.
While for $\aveTuc$, the monotonicity is not explicitly observed from Corollary 6. Due to space limitations, a sketch of the proof is given:
When $W$ increases, more TBs are allowed, thus more STBs occurs during the communication process, which also means shorter average update cycle.

When $W \rightarrow \infty$, the sensed information in each SB will be successfully transmitted to the sink, i.e., $\eventsuc$ always occurs and $\Psuc \rightarrow 1$.
Thus, UC contains the EHBs to harvest the amount of energy, $\Esen + \Etx$, the SB, and the blocks in $\Tstd$. Based on this explanation,
for the average update age, we have
\begin{equation}
\begin{aligned}
\lim\limits_{W\rightarrow \infty} \aveTstd
&\!=\!
\sum_{n=1}^{\infty}
\myprobability{N\!=\!n} \myexpect{\Tstd \vert N=n}
\!=\!
\sum_{n=1}^{\infty}
(1\!-\!\Pout) \Poutk{n-1} \myexpect{n+E((n-1)\Etx)} \\
&\!=\! \sum_{n=1}^{\infty}
(1-\Pout) \Poutk{n-1} \left(n+ (n-1) \frac{\Etx}{\rho}\right)
=
1+\left(\frac{\Etx}{\rho} +1\right) \frac{\Pout}{1-\Pout}.
\end{aligned}
\end{equation}
\par
\vspace{-0.4cm}
For the average update cycle, we have
\begin{equation}
\begin{aligned}
\lim\limits_{W\rightarrow \infty}
&\aveTuc
=
\sum_{n=1}^{\infty}
\myprobability{N=n} \myexpect{\Tuc \vert N=n}\\
&=
\sum_{n=1}^{\infty}
(1-\Pout) \Poutk{n-1} \myexpect{E(\Esen + \Etx)+1+n+E((n-1)\Etx)} \\
&=
\sum_{n=1}^{\infty}
(1\!-\!\Pout) \Poutk{n-1} \left(n\!+\!1 \!+\!  \frac{\Esen+n\Etx}{\rho}\right)
=
2\!+\!\left(\frac{\Etx}{\rho} \!+\!1\right) \frac{\Pout}{1-\Pout}\!+\! \frac{\Esen +\Etx}{\rho}.
\end{aligned}
\end{equation}


\ifCLASSOPTIONcaptionsoff
\fi


\end{document}